\documentclass[journal]{IEEEtran}

\usepackage{blindtext}
\usepackage{graphicx}
\usepackage[utf8]{inputenc}
\usepackage{array}
\usepackage{wrapfig}
\usepackage{multirow}
\usepackage{tabu}
\usepackage{algorithm, algpseudocode}
\usepackage{amsmath, amssymb}
\usepackage[utf8]{inputenc}
\usepackage[T1]{fontenc}
\usepackage{textcomp}
\usepackage{siunitx}
\usepackage{color}
\usepackage{physics}
\usepackage{cite}
\usepackage{hyperref}

\usepackage[dvipsnames]{xcolor}

\usepackage{amsthm}
\usepackage{physics}

\newtheorem{theorem}{Theorem}[section]
\newtheorem{lemma}[theorem]{Lemma}

\begin{document}
%----------------------------------------------------%
%                   Title
%----------------------------------------------------% 

\title{Identification of Stability Regions in Inverter-Based Microgrids}

\author{Andrey Gorbunov,~\IEEEmembership{Student Member,~IEEE}, Jimmy Chih-Hsien Peng,~\IEEEmembership{Senior Member,~IEEE},\\
Janusz W. Bialek,~\IEEEmembership{Fellow,~IEEE}, and Petr Vorobev,~\IEEEmembership{Member,~IEEE}%
%<-this % stops an unwanted space
%\thanks{ A. Gorbunov and J. C. -H. Peng are with the Department of Electrical \& Computer Engineering, National University of Singapore, Singapore (email: gorbunov@u.nus.edu, jpeng@nus.edu.sg)}
%\thanks{J.W. Bialek is with Skolkovo Institute of Science and Technology, Moscow, Russia, and with the School of Engineering, Newcastle University, Newcastle upon Tyne, U.K (email: j.bialek@skoltech.ru)}
%\thanks{P. Vorobev is with Skolkovo Institute of Science and Technology, Moscow, Russia. (email: p.vorobev@skoltech.ru)}
%\thanks{Work of P.Vorobev was supported by The Ministry of Education and Science of Russian Federation, Grant Agreement No 075-10-2021-067, Grant identification code 000000S707521QJX0002.}
}

% The paper headers
\markboth{}%
{Shell \MakeLowercase{\textit{et al.}}: Bare Demo of IEEEtran.cls for Journals}

\maketitle

\begin{abstract}
A new method for the stability assessment of inverter-based microgrids is presented in this paper. Directly determining stability boundaries by searching the multidimensional space of inverters' droop gains is a computationally prohibitive task. Instead, we build a certified stability region by utilizing a generalized Laplacian matrix eigenvalues, which are a measure of proximity to stability boundary. We establish an upper threshold for the eigenvalues that determines the stability boundary of the entire system and demonstrate that this value depends only on the network's R/X ratio but does not depend on the grid topology. We also provide a conservative upper threshold of the eigenvalues that are universal for any systems within a reasonable range of R/X ratios. We then construct approximate \emph{certified stability regions} representing convex sets in the multidimensional space of droop gains that could be utilized for gains optimization. We show how the certified stability region can be maximized by properly choosing droop gains, and we provide closed-form analytic expressions for the certified stability regions. The computational complexity of our method is almost independent of the number of inverters. The proposed methodology has been tested using IEEE 123 node test system with 10 inverters.
\end{abstract}


\IEEEpeerreviewmaketitle
%----------------------------------------------------%
%               I N T R O D U C T I O N 
%----------------------------------------------------% 
\section{Introduction}
%\cite{DepartmentofEnergy2011}
\IEEEPARstart{P}{ower} electronics interfaced generation is becoming increasingly widespread in modern power systems, mainly due to the increase in the share of renewable energy sources. Historically driven by governmental policies, this process is accelerating with the reduction of prices for semiconductor devices, and it is now being forecasted that
up to 80\% of all the generated electricity in the world will go
through power electronics devices by 2030 \cite{DepartmentofEnergy2011}. Thus, the so-called zero-inertia power systems---AC electric systems with no synchronous machines---have attracted a lot of attention from both research and industrial communities in recent years. Grid-forming inverters are the core elements of such systems and are now thought to become the main components of future power systems \cite{matevosyan2019grid,DepartmentofEnergy2011}. Unsurprisingly, inverter-based microgrids have attracted considerabl
e attention from the research community over the last years, with studies focused on different inverter control techniques, modeling approaches, security assessment, etc. The literature on the topic is vast, and the interested reader can refer to a review in \cite{parhizi2015state}.

Droop-controlled inverters mimic the behavior of synchronous machines \cite{Pogaku2007}, and are expected to become the main building blocks for zero-inertia grids. Power sharing and voltage regulation capabilities, the possibility of parallel operation, and relatively simple and universal operating principles make these inverters a promising solution for future power electronics-dominated grids. However, already the early experiments \cite{barklund2008energy} revealed that small-signal stability could become a major issue for such systems. Moreover, despite the seeming similarity of droop-controlled inverters to synchronous machines, the secure regions of inverter control parameters can be very narrow and require careful tuning of droop coefficients to guarantee the stable operation of such systems. Moreover, further research also revealed that modeling approaches that were routinely employed for conventional power systems fail to perform well for microgrids \cite{mariani2014model,vorobev2017high,Nikolakakos2016}. 

While direct detailed modeling can be used to certify stability of inverter-based microgrids \cite{naderi2019interconnected}, the computation cost of such an approach grows with the size of the system, and it becomes very high already for moderate-size grids with several inverters. In the past few years, there has been a lot of research in utilizing reduced-order models for simple and reliable stability assessment. In \cite{rasheduzzaman2015reduced}, a method is proposed to compute the coefficients of a reduced-order model numerically, which is sufficiently accurate to detect possible instabilities. However, such a model provides no insight into what parameters of the system influence stability the most. In \cite{Nikolakakos2016,nikolakakos2017reduced}, reduced-order models of different levels of complexity were analyzed, and the system states, critical for stability assessment, were determined. It was shown that contrary to conventional power systems, simple approaches based on time-scale separation do not perform well for microgrids. Moreover, \cite{Nikolakakos2016} showed that instability in inverter-based systems is mainly driven by the so-called "critical clusters" -  groups of inverters with short interconnection lines, and the heuristic method was proposed to detect these critical clusters. Later, in \cite{vorobev2017high}, and \cite{Vorobev2017}, the influence of critical clusters on system stability was confirmed by analytic models, although approximate ones. In \cite{gorbunov2020identification} the notion of critical clusters was defined using the highest eigenvalues of the generalized Laplacian matrix that we also exploit in this paper.

It is clear from the above that even stability certification of a given system configuration of a microgrid represents a rather complex problem. In practice, however, it is desirable to know the region of system parameters, where stable operation is guaranteed. Thus, one needs to consider \emph{stability regions} in the space of system parameters. If solved directly, this problem involves explicit verification of the system stability for multiple points in this space. Thus, in \cite{wu2014small} such point by point approach is used to find stability regions in a two-dimensional space of inverter droop coefficients (the ratio between droop coefficients of two inverters in the system is assumed to be constant). Likewise, in \cite{shuai2019parameter} and \cite{lenz2017bifurcation} bifurcation analysis is used to establish stability boundaries in the spaces of different pairs of system parameters. However, such direct approaches are only computationally feasible for building stability regions in low dimensional parameter space (like $2$-d space in the mentioned papers). Computational difficulty grows quickly (exponentially) with the number of parameters considered, which is a result of the \emph{curse of dimensionality}. Thus, constructing stability regions by point-by-point direct numerical assessment becomes infeasible for systems with multiple inverters.

To address this problem, in the present manuscript, we develop a methodology that allows determining approximate stability boundaries in the space of frequency and voltage droop gains for systems with an arbitrary number of inverters. The method does not require checking the stability of individual operating points within the space, and therefore it is not prone to the curse of dimensionality. The main original contributions of the paper are as follows:

\begin{enumerate}
    
    \item In Section \ref{sec:eigen} we generalize previously developed stability assessment \cite{gorbunov2020identification} method to account for arbitrary $R/X$ and droop coefficients ratios to provide conservative but easy to determine stability boundaries. The method is based on the analysis of the spectrum of the so-called generalized network Laplacian matrix, which in most cases has much fewer dimensions than the original state matrix.
    
    \item In Section \ref{sec:stability_regions}, we apply the developed methodology for constructing approximate \emph{stability regions} in the multidimensional space of droop gains for a system with multiple inverters. Namely, we provide \emph{closed form} analytic expressions for approximate stability regions, which represents a convex set (in the space of droop gains). This makes it very convenient for use in any optimization application.
    
    \item In Section \ref{sec:stability_regions} we show how the volume of the certified stability region can be maximized by properly choosing droop gains of the generalized Laplacian matrix. We also demonstrate that the computation cost of our method is almost independent of the number of inverters and network complexity.
\end{enumerate}

The methodology has been tested on a realistic model of the IEEE $123$ node system with $10$ inverters.

%----------------------------------------------------%
%               S E C T I O N  II
%----------------------------------------------------% 

\section{Modeling Approach} \label{sec:modeling}
In this section, we first recapitulate the dynamic model of inverter-based microgrids used for stability analysis. We adopt the $5$-th order electromagnetic model (EM) representing the grid-forming inverter as a droop-controlled AC voltage source, reported in many papers \cite{vorobev2017high,mariani2014model,machado2020network} as a suitable model for the stability assessment of inverter-based microgrids. The model consists of states associated with dynamics of each inverter (i.e., angles $\theta_i$, frequency $\omega_i$, and voltage $V_i$), each line (i.e., current components in \textit{dq} frame $I^{ij}_d$ and $I^{ij}_q$), and each load (i.e., current injections in $dq$ frame $[I_d]_i$ and $[I_q]_i$). Note that power, voltage, current, and impedances are given in per-unit. 
Denoting a set of lines as $\mathcal{E}$, a set of all nodes as $\mathcal{V}$, a set of buses with inverters as $\mathcal{V}_O$,  a set of virtual buses (nodes with zero current injections) as $\mathcal{V}_{virt}$, and a set of nodes with (passive) loads as $\mathcal{V}_L$, the dynamic model is represented by the following equations: 
\begin{subequations} \label{eq:dyn_nonlinear}
    \begin{align}
        &\dot{\theta_i} = \omega_i - \omega_0 \ , \ i \in \mathcal{V}_O,\\
        &\tau \dot{\omega_i} = \omega_{set} - \omega_i - \omega_0 m_i P_i \ , \ i \in \mathcal{V}_O,\\
        &\tau \dot{V_i} = V_{set} - V_i - n_i Q_i\ , \ i \in \mathcal{V}_O, \\
        &[I_d]_i = [I_q]_i = 0 \ , \ i \in \mathcal{V}_{virt} , \\
        &L_i \dot{[I_d]_{i}} = V_i \cos{\theta_i} - R_i [I_d]_i + \omega_0 L_i [I_q]_i \ , i\in\mathcal{V}_L\\
        &L_i \dot{[I_q]_i} = V_i \sin{\theta_i} - R_i [I_q]_i - \omega_0 L_i [I_d]_i \ , i\in\mathcal{V}_L , \\
        &L^{ij} \dot{I^{ij}_d} = V_i \cos{\theta_i}-V_j \cos{\theta_j} - R^{ij} I^{ij}_d + \omega_0 L^{ij} I^{ij}_q \ ,\\
        &L^{ij} \dot{I^{ij}_q} = V_i \sin{\theta_i}-V_j \sin{\theta_j} - R^{ij} I^{ij}_q - \omega_0 L^{ij} I^{ij}_d \ .
    \end{align}
\end{subequations}
Here, related to every inverter, $m_i$ and $n_i$ are frequency and voltage droop gains, respectively, and $\tau = \frac{1}{\omega_c}$ is the time constant of the power measurement low-pass filter \cite{Pogaku2007} with a cut-off frequency of $\omega_c$. $L^{ij}$ and $R^{ij}$ are inductance and resistance of the line connecting nodes $i$ and $j$, respectively, $L_i$ and $R_i$ are the inductance and the resistance of the loads. For brevity, the subscript of a variable represents the node index (e.g., $\theta_i, \ i \in \mathcal{V}$), and the superscript refers to those of an edge (e.g., $I_q^{ij}, \ (ij) \in \mathcal{E}$). Here, we use the constant impedance load representation. However, as was shown (both numerically and experimentally) in the work \cite{bottrell2013dynamic}, load dynamics have a little effect on the damping of low-frequency modes associated with inverter droop-controllers, which are the main focus of the present manuscript.

Next, the linear approximation of \eqref{eq:dyn_nonlinear} is formulated. The state variables are derived from the stationary operating point, such that $\delta \theta_i = \theta_i^{\text{actual}} - \theta_i^0$, where $\theta_i^{\text{actual}}$ is the actual angle value, and $\theta_i^0$ is the operational value. Here, we use the fact that angular differences between inverters are typically small in inverter-based microgrids. Hence, the operating values can be assumed as $\theta_i^0 \approx 0 \ \forall i\in \mathcal{V}$ and $V_i^0 \approx V_{set} = 1 \ p.u.$ \cite{mariani2014model,vorobev2017high}:

\begin{subequations} \label{eq:dyn_model}
    \begin{align}
        &\dot{\delta \theta_i} = \delta \omega_i \ , \ i \in \mathcal{V}_O,\\
        &\tau \dot{\delta \omega_i} = - \omega_i - \omega_0 m_i \delta P_i \  ,\ i \in \mathcal{V}_O,\\
        &\tau \delta\dot{ V_i} = -\delta V_i - n_i \delta Q_i\ , \ i \in \mathcal{V}_O,\\
        &[\delta I_d]_i = [\delta I_q]_i = 0 \ , \ i \in \mathcal{V}_{virt} , \\
        &L_i \dot{[\delta I_d]_{i}} = \delta V_i - R_i [\delta I_d]_i + \omega_0 L_i [\delta I_q]_i \ , i\in\mathcal{V}_L\\
        &L_i \dot{[\delta I_q]_i} = \delta \theta_i - R_i [\delta I_q]_i - \omega_0 L_i [\delta I_d]_i \ , i\in\mathcal{V}_L , \\
        &L^{ij} \dot{\delta I^{ij}_d} = \delta V_i - \delta V_j - R^{ij} \delta I^{ij}_d + \omega_0 L^{ij} \delta I^{ij}_q \ \label{eq:dyn_Id}, \\
        &L^{ij} \dot{\delta I^{ij}_q} = \delta \theta_i - \delta \theta_j - R^{ij} \delta I^{ij}_q - \omega_0 L^{ij} \delta I^{ij}_d \label{eq:dyn_Iq}  \  .
    \end{align}
\end{subequations}
To simplify the notations, the state variables associated with deviations, e.g., $\delta \theta_i, \delta V_i$, etc., will be denoted merely as $\theta_i, V_i, \cdots$ from the hereafter. 

Usually, the number of lines is much higher than the number of buses, so the total number of equations in the system \eqref{eq:dyn_model} can become big even for moderate-size grids. The standard approach, used for conventional power systems, utilizes the well-pronounced time-scale separation between electro-mechanical and electromagnetic phenomena so that that line dynamics can be neglected (i.e.,, $\dot{I_d}, \dot{I_q} \approx 0$) when considering generator angles dynamics. This allows to introduce nodal currents that are related to nodal voltages via the admittance matrix $\pmb{I} = Y \pmb{V}$ (i.e., by a set of algebraic relations). However, for inverter-based microgrids, explicit accounting for line dynamics is crucial for small-signal stability analysis, and derivative terms in equations \eqref{eq:dyn_Id} and \eqref{eq:dyn_Iq} cannot be neglected \cite{vorobev2017high,Nikolakakos2016}. Nevertheless, a nodal representation of \eqref{eq:dyn_model} and, more importantly, Kron reduction is still possible under uniform $R/X$ assumption \cite{tucci2016plug}, \cite{caliskan2012kron}, and will be discussed in the following subsection.  
 
\subsection{Nodal representation of the network dynamics}
 
Let us denote the incidence matrix of the power system network as $\nabla^T$ with a size of $|\mathcal{V}|\times |\mathcal{E}|$, such that $\nabla^T_{ij} = -1$ if the $j$th line leaves the $i$th node, or $\nabla^T_{ij} = 1$ if the $j$th line enters the $i$th node we remind, that the nodes can be inverter nodes, load nodes, or virtual nodes.
The electromagnetic dynamics of the lines from equations \eqref{eq:dyn_Id} and \eqref{eq:dyn_Iq}  can then be represented by the following vector equations:
\begin{subequations} \label{eq:Ohms_law}
    \begin{align}
        &\pmb{\dot{\mathcal{I}}_d} = \omega_0 X^{-1}\nabla \pmb{V}_a - \omega_0 \mathcal{P} \pmb{\mathcal{I}_d} + \omega_0 \pmb{\mathcal{I}_q} \ ,\\
        &\pmb{\dot{\mathcal{I}}_q} = \omega_0 X^{-1}\nabla \pmb{\theta}_a - \omega_0 \mathcal{P}  \pmb{\mathcal{I}_q} - \omega_0 \pmb{\mathcal{I}_d} \ ,
    \end{align}
\end{subequations}
where $X = \text{diag}(\cdots,X^{ij},\cdots)$ is a matrix with line reactances on its diagonal; $\mathcal{P} =\text{ diag}(\cdots,\rho^{ij},\cdots)$ is a matrix with $R/X$ ratios on its diagonal; $\pmb{V}_a$ and $\pmb{\theta}_a$ are vectors consisting of voltages and phase angles at each bus in $\mathcal{V}$; $\pmb{\mathcal{I}_d}$ and $\pmb{\mathcal{I}_q}$ are vectors of currents associated with each line in $ \mathcal{E}$. One should notice that the vectors $\pmb{V}_a, \pmb{\theta}_a$ and vectors $\pmb{\mathcal{I}_d}, \pmb{\mathcal{I}_q}$ have different dimensions: $|\mathcal{V}|$ and $|\mathcal{E}|$ accordingly. Let us multiply \eqref{eq:Ohms_law} by $\nabla^T$, and introduce the denotations $\{\pmb{I_d}\}_a = \nabla^T \pmb{\mathcal{I}_d}, \ \{\pmb{I_q}\}_a = \nabla^T \pmb{\mathcal{I}_q}$ for nodal currents, such that: 
\begin{subequations} \label{eq:EM_dynamics}
    \begin{align}
        &\{\pmb{\dot{I}_d}\}_a = \omega_0 \nabla^T X^{-1}\nabla \pmb{V}_a - \omega_0 \nabla^T \mathcal{P}  \pmb{\mathcal{I}_d} + \omega_0 \{\pmb{I_q}\}_a \ ,\\
        &\{\pmb{\dot{I}_q}\}_a = \omega_0 \nabla^T X^{-1}\nabla \{\pmb{\theta}\}_a - \omega_0 \nabla^T \mathcal{P}  \pmb{\mathcal{I}_q} - \omega_0 \{\pmb{I_d}\}_a \ .
    \end{align}
\end{subequations}
Equations \eqref{eq:EM_dynamics} are almost in the desired form with the vectors $\pmb{V}_a$, $\pmb{\theta}_a$ and $\{\pmb{I_d}\}_a, \{\pmb{I_q}\}_a$ being nodal, except for the presence of the term $\nabla^T \mathcal{P}  \pmb{\mathcal{I}}_d$. 

 Let us consider a useful case when all the lines have the same R/X ratio, $\mathcal{P} = \rho \mathbf{1}$ (homogeneous $\rho$). In this case, $\nabla^T \mathcal{P}  \pmb{\mathcal{I}}_d = \nabla^T \rho \pmb{\mathcal{I}}_d = \rho \{\pmb{I_d}\}_a$ and the nodal representation is \cite{tucci2016plug}, \cite{caliskan2012kron}:
\begin{subequations} \label{eq:EM_dynamics_homogen_full}
    \begin{align}
        &\{\pmb{\dot{I}_d}\}_a = \omega_0 (1+\rho^2)B_a \pmb{V}_a - \omega_0 \rho \{\pmb{I_d}\}_a + \omega_0 \{\pmb{I_q}\}_a \ ,\\
        &\{\pmb{\dot{I}_q}\}_a = \omega_0 (1+\rho^2)B_a \pmb{\theta}_a - \omega_0 \rho  \{\pmb{I_q}\}_a - \omega_0 \{\pmb{I_d}\}_a ,
    \end{align}
\end{subequations}
where $B_a = -\Im{Y}$ is the nodal susceptance matrix of all nodes in $\mathcal{V}$.

Apart from reducing the number of equations in \eqref{eq:Ohms_law} for a system with more lines than nodes, another advantage of \eqref{eq:EM_dynamics_homogen_full} is the fact that Kron reduction procedure can be performed over it, which eliminates nodes with zero injections $\{\pmb{I_q}\}_0=\{\pmb{I_q}\}_0 = 0$ as follows,
\begin{subequations} \label{eq:EM_dynamics_homogen}
    \begin{align}
        &\{\pmb{\dot{I}_d}\}_1 = \omega_0 (1+\rho^2)B_1 \pmb{V}_1 - \omega_0 \rho \{\pmb{I_d}\}_1 + \omega_0 \{\pmb{I_q}\}_1 \ ,\\
        &\{\pmb{\dot{I}_q}\}_1 = \omega_0 (1+\rho^2)B_1 \pmb{\theta}_1 - \omega_0 \rho  \{\pmb{I_q}\}_1 - \omega_0 \{\pmb{I_d}\}_1 ,
    \end{align}
\end{subequations}
where the susceptance matrix $B_1$ and vectors $\pmb{V}_1, \pmb{\theta}_1, \{\pmb{I_d}\}_1, \{\pmb{I_q}\}_1$ correspond to nodes with eliminated virtual buses. The derived nodal representation is used in the next subsection to get the state-space model.

\subsection{The state-space model with homogeneous $\rho$} \label{subsec:same_rx}

To represent the dynamic model in the state-space form, we first express inverters' real and reactive power in terms of their voltage and current. Again, assuming $V_i^0 \approx  1$ p.u. and $\theta_i^0 \approx 0$, the relationship between nodal power injections $\mathbf{P}, \mathbf{Q}$ and nodal currents $\pmb{I}_d, \pmb{I}_q$ and the inverter buses can be expressed as follows:

\begin{equation}
    \begin{pmatrix}
        \mathbf{P} \\
        \mathbf{Q} 
    \end{pmatrix} = 
    \begin{bmatrix}
        \mathbf{1} & 0\\
        0 & -\mathbf{1}\\
    \end{bmatrix}
    \begin{pmatrix}
        \pmb{I}_d\\
        \pmb{I}_q
    \end{pmatrix} .
\end{equation} 
Next, using the fact that for microgrids, the load impedance is much higher than the network impedances (see, e.g., \cite{vorobev2017high,Pogaku2007}), one can show that load buses can be excluded similarly to the virtual ones (see detailed discussion and derivation in the Appendix \ref{app:elimination_loads}). Subsequently, using \eqref{eq:EM_dynamics_homogen} the following state-space representation of the electromagnetic (EM) $5^{th}$ order model of a microgrid with the homogeneous $\rho$ is obtained:

\begin{equation} \label{eq:5th_model}
    \begin{pmatrix}
        \pmb{\dot{\theta}}\\
        \tau\pmb{\dot{\omega}}\\
        \tau\pmb{\dot{V}}\\
        \tau_0\pmb{\dot{I}}_d\\
        \tau_0\pmb{\dot{I}}_q
    \end{pmatrix} =
    \begin{bmatrix}
        0 & \mathbf{1} & 0 & 0 & 0 \\
        0 & -\mathbf{1} & 0 & -\omega_0 M & 0 \\
        0 & 0 & -\mathbf{1} & 0 & N\\
        0 & 0 & \mathcal{B} & -\rho \mathbf{1} & \mathbf{1} \\
        \mathcal{B} & 0 & 0 & -\mathbf{1} & -\rho \mathbf{1}
    \end{bmatrix}
    \begin{pmatrix}
        \pmb{\theta}\\
        \pmb{\omega}\\
        \pmb{V}\\
        \pmb{I}_d\\
        \pmb{I}_q
    \end{pmatrix} \ ,
\end{equation}
where $M = \text{diag}(m_1, \cdots, m_v), N = \text{diag}(n_1, \cdots, n_v)$ are diagonal matrices of droop gains for each inverter ($v = |\mathcal{V}_O|$ is the number of inverters into the system), $\mathcal{B} = (1+\rho^2) B$ and $B$ is the susceptance matrix of the Kron-reduced system, and all vectors $\pmb{V}, \pmb{\theta}, \pmb{\omega}, \pmb{I_d}, \pmb{I_q}$ also correspond to the inverters nodes such that each sub-matrix in \eqref{eq:5th_model} has a size of $v\times v$. Note that \eqref{eq:5th_model} does not include load admittance, i.e., the subsequent theoretical analysis assumes the unloaded system (see Appendix \ref{app:elimination_loads} for more details). We note, however, that in subsequent numerical simulations, we use the system with loads to verify and validate our result. 
%----------------------------------------------------%
%               S E C T I O N  III
%----------------------------------------------------% 
\section{Decomposition Into Two-bus Equivalents}\label{sec:eigen}

This section provides a theoretical foundation for our method of identification of stability regions. The fundamental results, previously reported by us in \cite{gorbunov2020identification} (i.e., Theorem \ref{theorem:one} and partially Theorem \ref{theorem_addition_line} are also included here for completeness.

The model in \eqref{eq:5th_model} can be written in the state-space form, i.e., $\dot{\pmb{x}}= A\pmb x$, such that the stability analysis of the system \eqref{eq:dyn_model} can be evaluated by the eigenvalue analysis of the matrix $A$. In this case, the state matrix in \eqref{eq:5th_model} exhibits a specific structure---a five-by-five block matrix where each $v\times v$ sub-matrix is symmetric---a property that we will exploit.

To derive the stability criteria in a compact form, let us assume that the frequency and voltage power droop coefficients ratio is uniform across the system, i.e., in matrix terms $M = k N$. In this case, the dynamic model \eqref{eq:5th_model} can be expressed in terms of $\pmb{\theta}$ in the Laplace domain in the following compact way \cite{gorbunov2020identification},
\begin{equation} \label{eq:theta_repr}
    [\tau k f(s) s I + g(s) (k + \tau s) M \mathcal{B} + (M \mathcal{B})^2] \pmb{\theta} = 0 \ ,
\end{equation}
where $g(s) = (1 + \tau s)$, $f(s) = g^2(s) [(\rho + \frac{s}{\omega_0})^2 + 1]$. Subsequently, the following theorem is formulated \cite[Theorem~III.1]{gorbunov2020identification}, which forms the basis for our stability regions identification method: 
\begin{theorem} \label{theorem:one}
    The eigenvalues $\lambda$ of \eqref{eq:5th_model}, under proportional droops $M = kN$ assumption, are related to the eigenvalues $\{\mu_i, \ i = 1,\cdots, v\}$ of the generalized Laplacian matrix $C = M\mathcal{B}$ by the following algebraic equations:
    \begin{equation} \label{eq:two_bus_polynomial}
        \tau kf(\lambda) + g(\lambda)(k+\tau \lambda) \mu_i + \mu_i^2 = 0 .
    \end{equation}
    
\end{theorem}
Thereby, the stability analysis of the system in \eqref{eq:5th_model} reduces to the analysis of polynomials in $\lambda$ from \eqref{eq:two_bus_polynomial} for all values of $\mu_i$. It could be shown that each polynomial in \eqref{eq:two_bus_polynomial} (i.e., fifth-order polynomial in $\lambda$ with a specific $\mu_i$)
coincides with the characteristic polynomial of an inverter versus an infinite bus system with the following parameters: $m^{eq}=kn^{eq}, R^{eq}=\rho X^{eq}$ and $\mu_i = \frac{m^{eq}}{X^{eq}}$. Thus, the initial microgrid can be effectively split into separate two-bus equivalents utilizing the eigenvalues of the generalized Laplacian matrix $C = M\mathcal{B}$. Fig. \ref{fig:kundur_decoupling} illustrates such a split for a two-area four-inverter system that is analogous to the Kundur's two-area system \cite{klein1991fundamental}. Fig. \ref{fig:kundur_eigplot_load_var} compares the dominant eigenvalues of two-bus equivalents \eqref{eq:two_bus_polynomial} and of the full model \eqref{eq:dyn_model}. The simulation code has been published online \footnote{\href{https://github.com/goriand/Identification-of-Stability-Regions-in-Inverter-Based-Microgrids}{https://github.com/goriand/Identification-of-Stability-Regions-in-Inverter-Based-Microgrids}} As the equation \eqref{eq:two_bus_polynomial} has been derived assuming a no-loaded system, for the validation over the system \eqref{eq:dyn_model} we have varied the loads in the range of $0.5-2$ p.u. to check if it makes much difference. The two-bus equivalents \eqref{eq:two_bus_polynomial} correspond to the following colors: green for Inter-Area mode, blue for Area I, and red for Area II, while purple dots show $100$ samples of eigenvalues of the full model \eqref{eq:dyn_model} with random load variations. One can observe that the eigenvalues of two-bus equivalents are very close to the eigenvalues of \eqref{eq:dyn_model} and the load variation has very little effect on the eigenvalues which confirms the derivations in Appendix \ref{app:elimination_loads}.

\begin{figure}
    \centering
    \includegraphics[width = 0.30\textwidth]{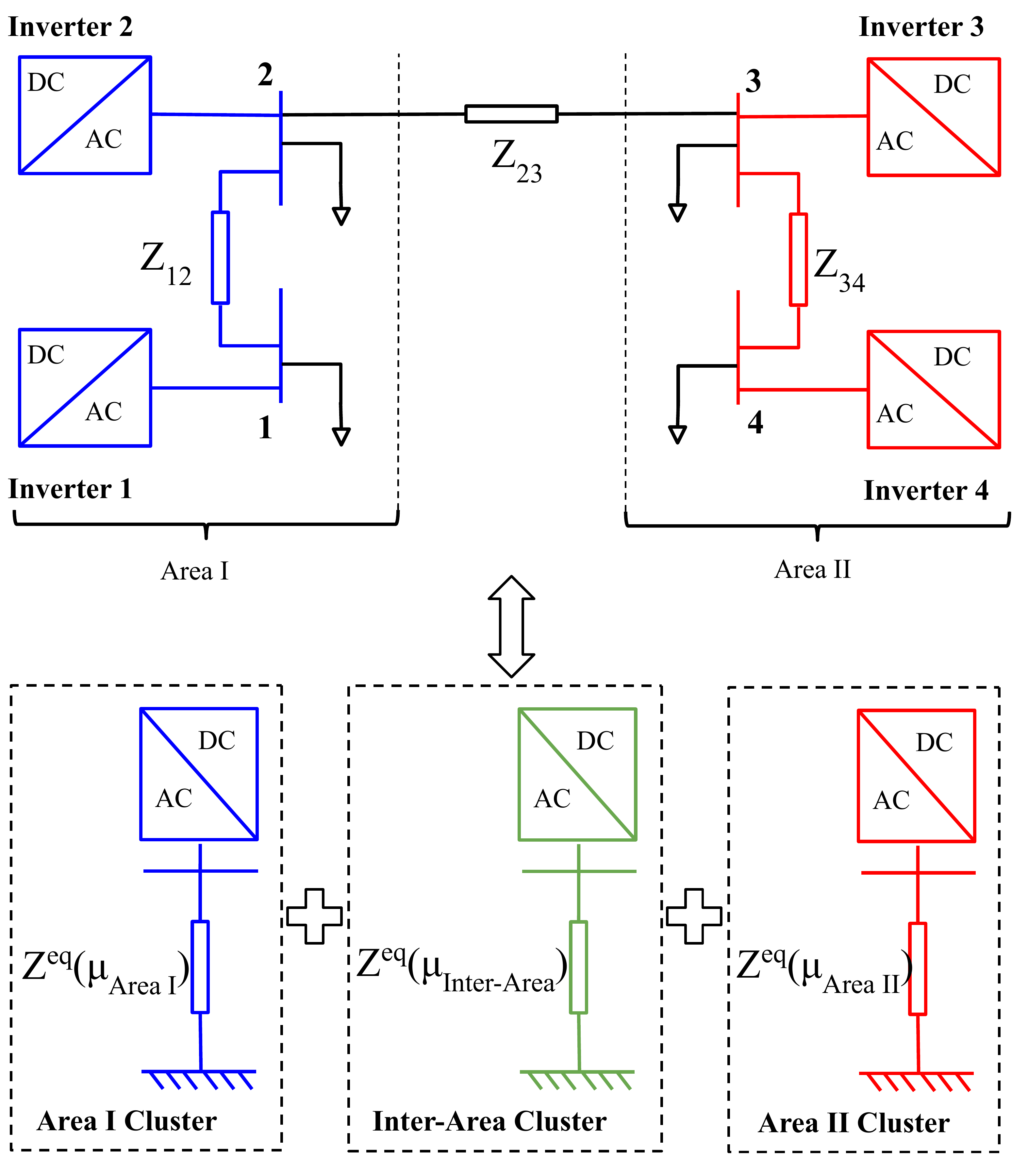}
    \caption{Illustration of decomposition of a system into a set of clusters - equivalent two-bus systems.}
    \label{fig:kundur_decoupling}
\end{figure}

\begin{figure}
    \centering
    \includegraphics[width=0.35\textwidth]{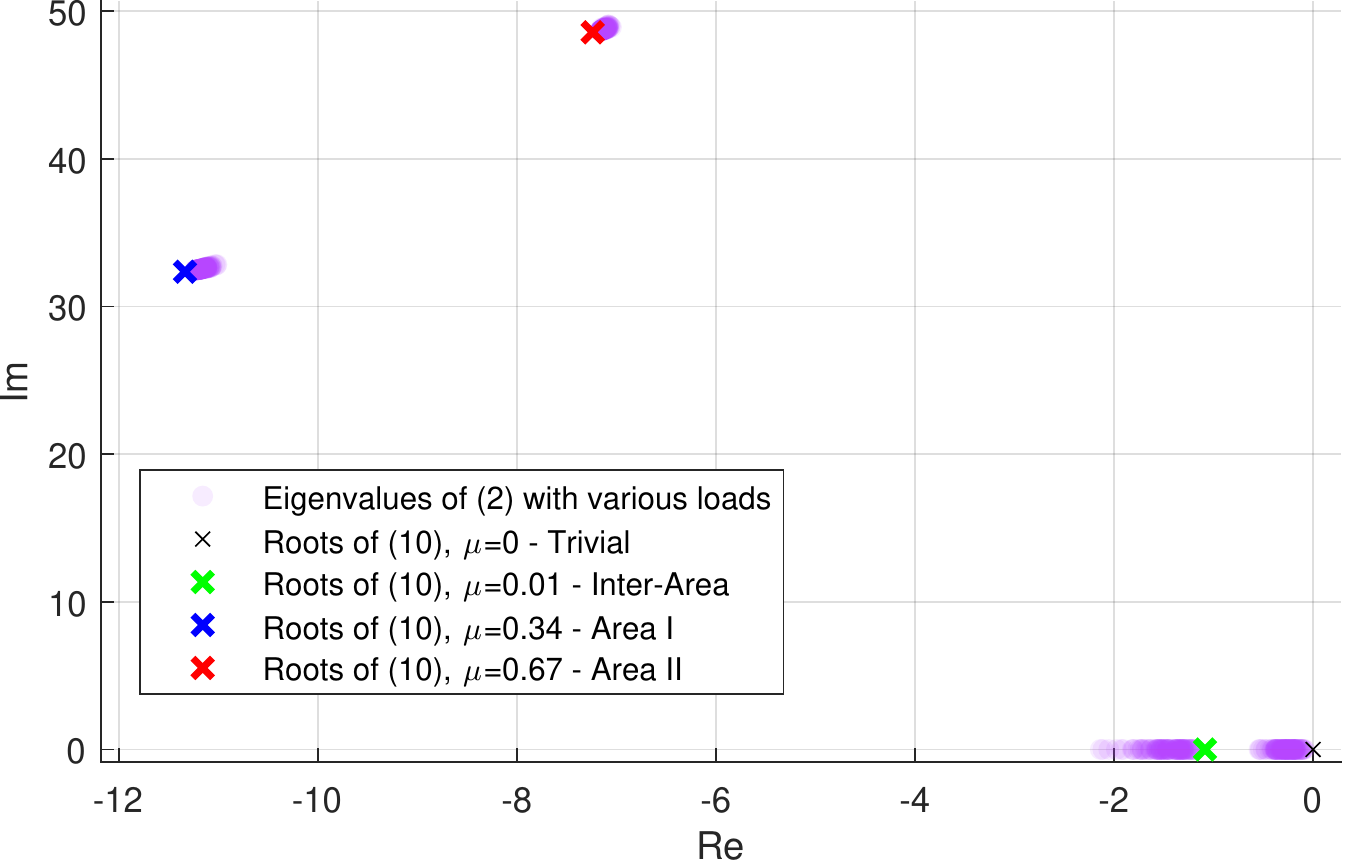}
    \caption{Eigenvalue plot for the two-area system of Fig. \ref{fig:kundur_decoupling} with varying load impedances. Only the dominant modes of the upper-half plane are shown.}
    \label{fig:kundur_eigplot_load_var}
\end{figure}

Each equivalent two-bus system in Fig. \ref{fig:kundur_decoupling} corresponds to one value of $\mu_i$. Thus, the initial grid's stability assessment is reduced to a set of equivalent two-bus systems.
An important property of the system in \eqref{eq:two_bus_polynomial} is that the value of $\mu_i$ quantifies the system stability, i.e.,, the higher $\mu_i$, the less stable the system is \cite{gorbunov2020identification}. For the case of two-bus equivalent systems, this can be explained by the fact that such system is less stable for higher values of $1/X^{eq}$ and droop gains $m_i$, as demonstrated in \cite{Pogaku2007}, \cite{vorobev2017high}. Therefore, a two-bus system's stability margin is also less for higher values of $\mu_i=m^{eq}/X^{eq}$. Thus, there exist an upper boundary $\mu_{cr}$ such that if for an initial system all $\mu_i < \mu_{cr}$, then the system is stable, and vice versa---if there are instances where $\mu_i > \mu_{cr}, \ i=1,\cdots,u$, then the system has $u$ unstable equivalent two-bus systems \cite{gorbunov2020identification}.

It is noteworthy that the critical value $\mu_{cr}$ is determined only by \eqref{eq:two_bus_polynomial} and is unique for a whole family of systems with particular uniform  $\rho=R/X$ ratio and droop gain ratios $k$. Therefore, $\mu_{cr}$ does not depend on a given system's topology. The value of $\mu_{cr}$ can be found numerically for every pair of $\rho$ and $k$. Thus, one can say that a certain value of $\mu_{cr}$ corresponds to a class of networks. The map of $\mu_{cr}$ as a function of $\rho$ for different values of $k$ in a practical range of $\rho \in [0.4,5]$ and $k \in [0.3;5]$ is illustrated in Fig. \ref{fig:mu_cr_map_proj} using the values $\tau = 1/10\pi, \omega_0 = 100\pi$. 

\begin{figure}
    \centering
    \includegraphics[width=0.4\textwidth]{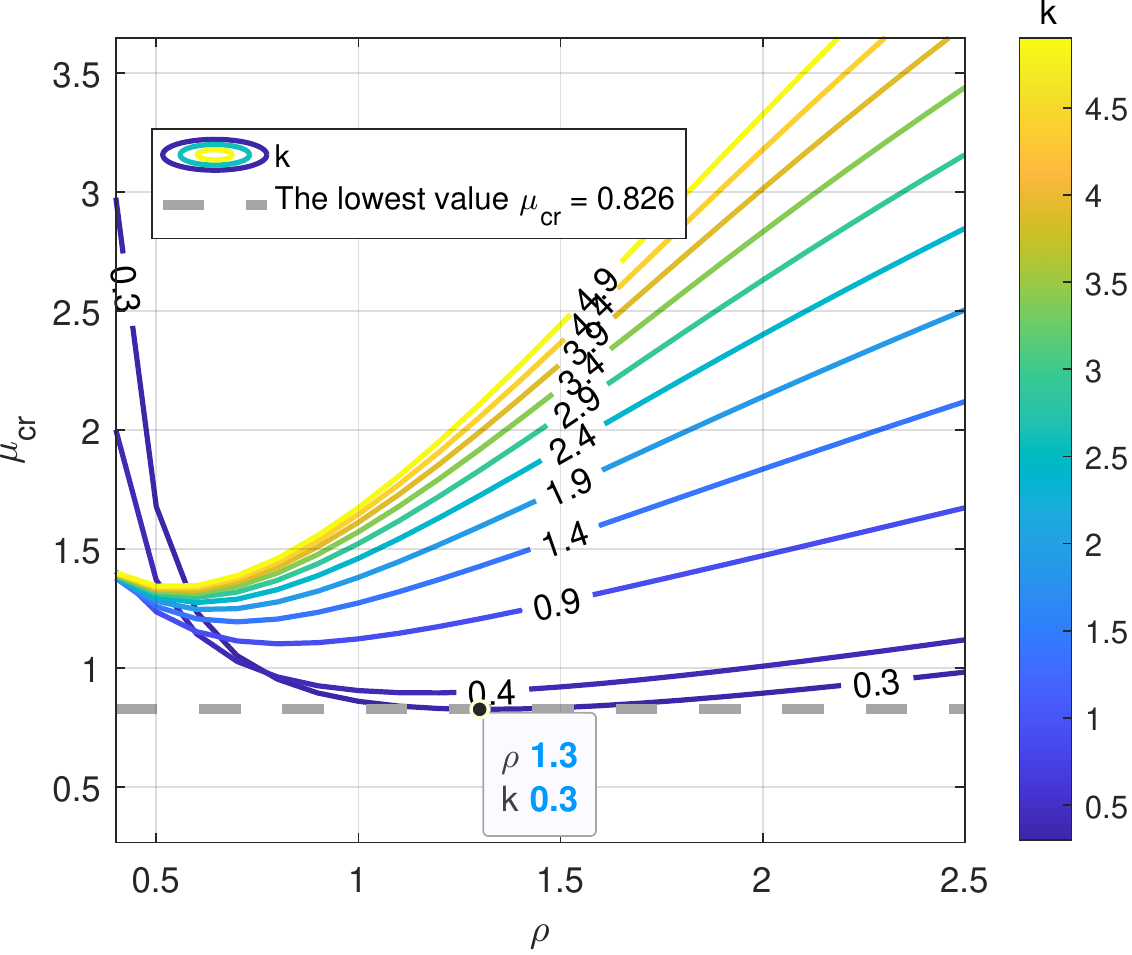}
    \caption{
    Upper threshold $\mu_{cr}$ with various parameters $\rho$ and $k$. The \textit{worst-case} (i.e., the smallest upper bound $\mu_{cr}$ in both $k$ and $\rho$) corresponds to $k=0.3, \rho=1.3$, where $\mu_{cr,min} = 0.826$.}
    \label{fig:mu_cr_map_proj}
\end{figure}

Now we will prove that the addition of any line and the increase of droop gain leads to an increase of $\mu_i$ in a system with multiple inverters. This property is a generalization of the observation for a two-bus system \cite{Pogaku2007,vorobev2017high}.

%----------------------------------------------------%
%               S E C T I O N  VI
%----------------------------------------------------% 

\subsection{Properties of the eigenvalues of $C$} \label{sec:sensitivity}

\begin{figure}
    \centering
    \includegraphics[width=0.35\textwidth]{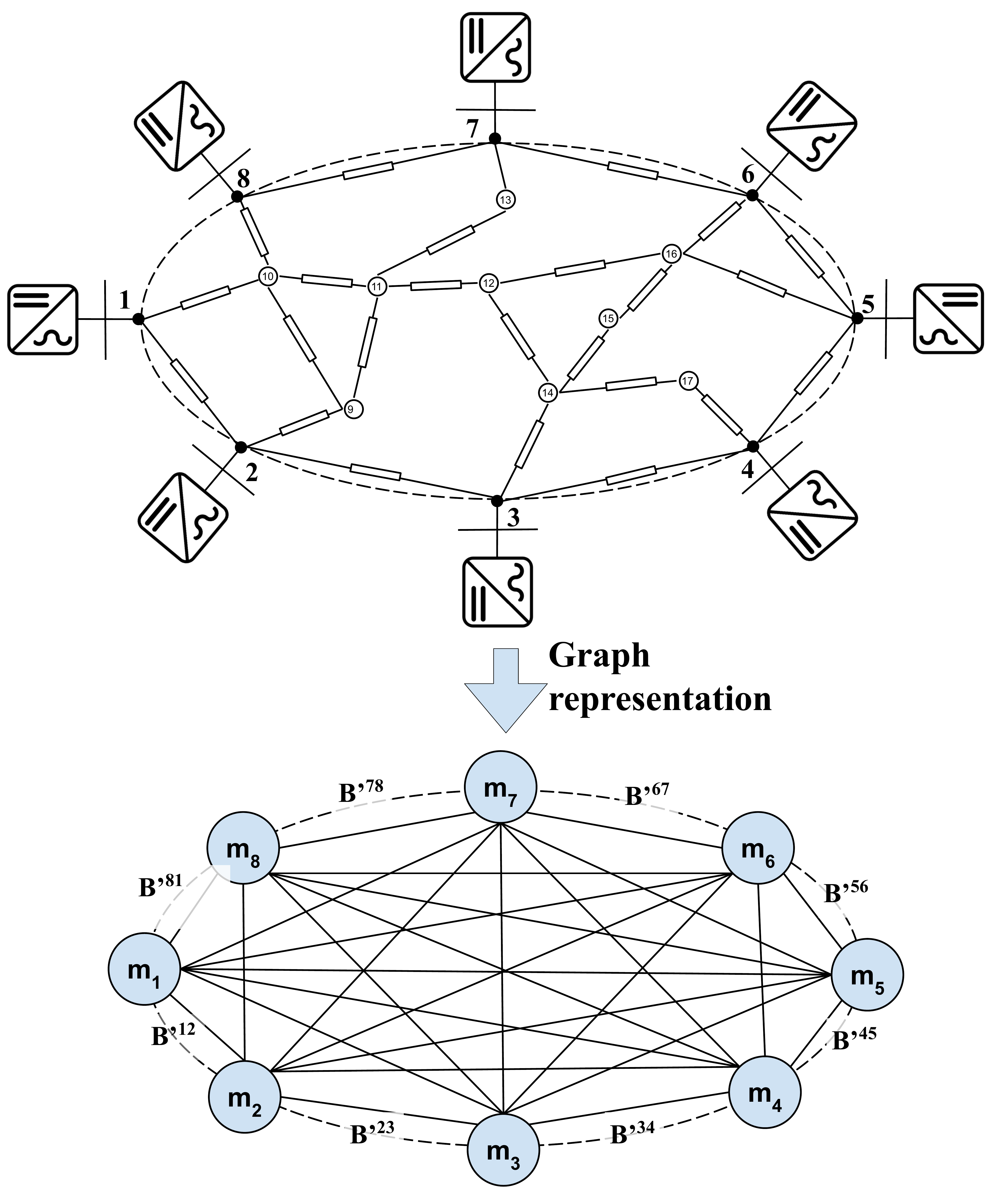}
    \caption{Representation of an inverter-based microgrid by a graph with $m_i$ as node weights and $\mathcal{B}^e$ as branch weights obtained after Kron-reduction (labels for some lines are omitted).}
    \label{fig:graph_representation}
\end{figure}

Matrix $C$ could be considered as the generalized Laplacian matrix \cite{godsil2001laplacian} for the Kron-reduced network graph augmented by node weights equal to droop gains $m_i, \ i=1,\cdots,l$ as illustrated in Fig. \ref{fig:graph_representation}. We note that all the eigenvalues of $C$ are real and  non-negative. Although $C$ is not symmetric, it can be made symmetric by a similarity transform  $M^{-1/2}C M^{1/2} = M^{1/2} \mathcal{B} M^{1/2}$.
Since $M$ is a diagonal matrix, calculating $M^{1/2}$ is trivial and proves that the eigenvalues of $C$ are real. Moreover, both $M^{1/2}$ and $\mathcal{B}$ matrices are positive (semi)definite, so is the product of $M^{1/2} \mathcal{B} M^{1/2}$, which proves that the eigenvalues of $C$ are non-negative. As the matrix $C$ involves dimensionless values of $1/X^{ij}$ (in p.u.) and $m_k$ (in \%), its eigenvalues $\mu_i$ are also dimensionless.  

Having proved that the eigenvalues of $C$ are non-negative, we can formulate the following theorem, which establishes important stability properties of inverter-based microgrids:
\begin{theorem}\label{theorem_addition_line}
    The addition of any new line, as well as the increase of any existing line susceptance $\mathcal{B}^e = \frac{1}{X^e}$, or the increase of any inverter's droop gain $m_k, k = 1, \cdots, v$ in the system could only increase eigenvalues $\mu_i, \ i=1,\cdots, v$.
\end{theorem}
\begin{proof}

The proof for line addition can be found in \cite[Theorem~III.2]{gorbunov2020identification}.
Here, we show that $\mu_i$ increases with the droop gains, $m_k$. Let us represent the increase of $m_k$ by its multiplication by some value $d_k>1$. In this case, we relate original $C$ with $\tilde{C}$ for the increased droop $d_k m_k$ as $\tilde{C} = D C$, where $D = \text{diag}(1, \cdots, d_k, \cdots, 1)$. Next, we use a multiplicative version of the Weyl's inequality \cite[Theorem~4.1]{so1994commutativity}:
    \begin{equation}
        \mu_i \leq \tilde{\mu}_i \leq  d_k\mu_i, \ i=1,\cdots, v \ ,
    \end{equation}
    where $\tilde{\mu}_i$ are eigenvalues of $\tilde{C}$.
\end{proof}

Therefore, Theorem \ref{theorem_addition_line} suggests that the addition of a new line or increase of the susceptance of any existing line makes an inverter-based microgrid less stable. The same is also true for an increase in any inverter droop gain. The property is entirely consistent with the previous results in \cite{vorobev2017high}, \cite{vorobev2018towards}. This property is a distinctive feature of inverter-based microgrids in comparison with conventional power systems.

\subsection{Worst-case droop gains and R/X ratios} \label{sec:identification_procedure}
In order to generalize our approach to networks with non-uniform $R/X$ ratios we note that Fig. \ref{fig:mu_cr_map_proj} suggests that there exist a certain \textit{worst-case} scenario, characterised by certain values of $\rho_{ext}, k_{ext}$, and $\mu_{cr,min}$ (specifically,  $\rho_{ext}=1.3$, $k_{ext}=0.3$, and $\mu_{cr,min}=0.826$ from Fig. \ref{fig:mu_cr_map_proj}) such that any deviation from $\rho_{ext}$ and $k_{ext}$ makes the system more stable (i.e., increases the stability boundary $\mu_{cr}$). 
For example, Fig. \ref{fig:kundur_eigplot_k_rho_var} illustrates the conservative stability boundary with worst-case scenario for the two-area system of Fig. \ref{fig:kundur_decoupling}. Namely, red crosses are the eigenvalues of \eqref{eq:dyn_model} with $R/X = \rho_{ext}, m_i/n_i=k_{ext}$ and droop gains chosen such that the system is marginally stable (i.e., $\mu_v = \mu_{cr,min} = 0.826$ of Fig. \ref{fig:mu_cr_map_proj}), while purple dots are eigenvalues corresponding to $10000$ random samples of $R/X \in [0.4, 2.5]$ and $k=m_i/n_i \in [0.3,5]$ for the same system. All the purple dots are lying in the left-half plane with a higher stability margin than the marginal eigenvalue, which therefore confirms the system with $k_{ext}, \rho_{ext}$ is the least stable one.

Moreover, we prove the following local stationarity conditions for non-uniform $R/X$, indicating the theoretical existence of such worst case.

\begin{theorem}
    Assuming uniform droop ratio $k=m/n$, the stationary value $\rho_{ext}$ for the system with uniform $R/X$ \eqref{eq:5th_model} is also the stationary value for the system \eqref{eq:dyn_model} with non-uniform $R/X$, i.e., if $\frac{\partial \Re(\lambda)}{\partial \rho}\Bigr|_{\rho_{ext}} = 0$, then $\frac{\partial \Re(\lambda)}{\partial \rho_j}\Bigr|_{\rho_{ext}} = 0, \ \forall j$,
    where $\rho_j$ is the $R/X$ ratio of the $j$-th line.
\end{theorem}
\begin{proof}
    This proof is based on the following representation $\frac{\partial\lambda}{\partial \rho_j} = \alpha_j f(\rho)$ with some real constants $\alpha_j$ and complex function $f(\rho)$, which is the same for each $j$. If the above representation is obtained, then indeed all $\frac{\partial\Re{\lambda}}{\partial \rho_j}\Bigr|_{\rho_{cr}} = 0$ whenever $\Re(f(\rho_{cr})) = 0$. An abbreviated version of the proof is in Appendix \ref{app:stationarity_proof}.
    %is given in the Appendix. 
\end{proof}

While the condition $\mu_i=\mu_{cr}$ for the eigenvalues of $C$ provides the exact stability boundary only for the system with uniform $R/X$ and $m/n$, matrix $C$ could still be constructed even for a system with non-uniform parameters. Indeed, by definition, $C$ is constructed using only the values of line reactances $X$ and inverter frequency droop gains $m_i$, and not taking into account line resistances $R$ and inverter voltage droops $n_i$. %However, the order of $C$'s eigenvalues, i.e., $\mu_i$, would not necessarily correspond to the order of the proximities to the stability boundary. 
Moreover, this matrix (more precisely, its eigenvalues) can be used to get conservative stability boundaries for systems with arbitrary network parameters. As was shown in the previous paragraph, we can define the \textit{worst-case} threshold $\mu_{cr, min}$ such that all systems satisfying:
\begin{equation}\label{eq:worst_case}
    \mu_i(C) < \mu_{cr, min},\ \forall i \ 
\end{equation}
are definitely stable. However, in the case of non-uniform $R/X$, violation of condition \eqref{eq:worst_case} does not necessarily mean that the system is unstable. The value $\mu_{cr,min}$ equals to the minimum value of $\mu_{cr}$ across a reasonable range of $\rho \in [0.4, 2.5]$ and $k \in [0.3,5]$, as demonstrated in Fig. \ref{fig:mu_cr_map_proj}, where $\mu_{cr,min} = 0.826$ takes its minimum at $\rho_{ext} = 1.3, k_{ext} = 0.3$. We note the threshold $\mu_{cr,min}$ is universal for any system with parameters in the given range of $\rho,k$. 

\begin{figure}
    \centering
    \includegraphics[width=0.4\textwidth]{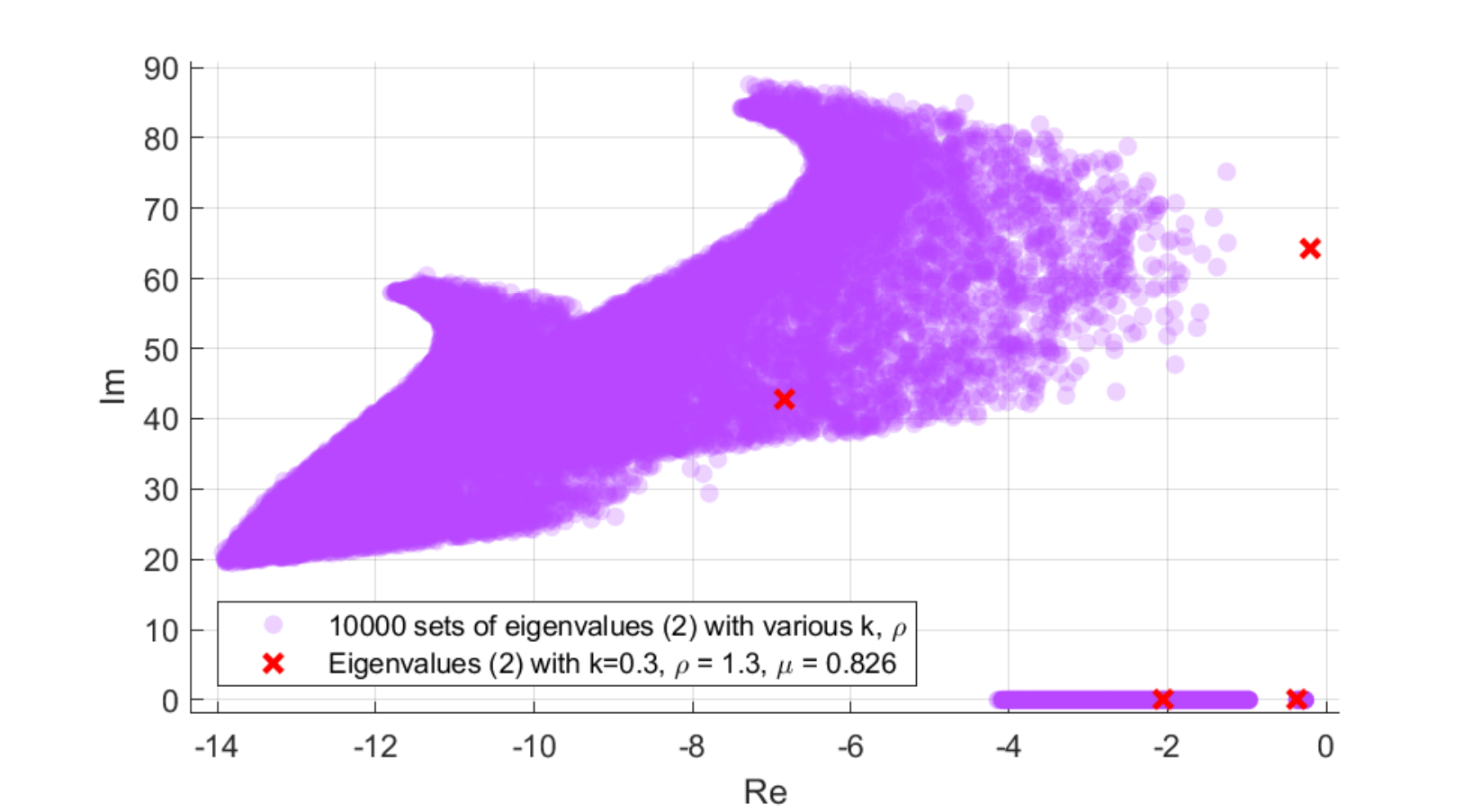}
    \caption{Eigenvalue plot for the two-area system of Fig. \ref{fig:kundur_decoupling} with various (non-uniform) $R/X$ and $M/N$ parameters. Only the dominant modes of the upper-half plane are shown.}
    \label{fig:kundur_eigplot_k_rho_var}
\end{figure}

%----------------------------------------------------%
%               S E C T I O N  V
%----------------------------------------------------% 
\section{Stability Regions of Individual Inverters} \label{sec:stability_regions}

In this section, approximate (conservative) \emph{stability region} in the multidimensional space of droop gains $m_i$ and $n_i$ of all the system inverters will be derived using $\mu_{cr,min}$ \eqref{eq:worst_case}. Such stability regions are valuable for practical applications and can be utilized, for example, to optimize the values of gains concerning some criteria (e.g., in \cite{dorfler2015breaking}). In order to make the value of our stability assessment method clear, we note that the exact stability boundary for a multi-inverter microgrid is a complicated hyper-surface in the $2v$-dimensional space ($v$ is the number of inverters in the system) of the droop gains. While direct numerical analysis can be used to certify the stability of the system \eqref{eq:dyn_model} under fixed values of all the parameters, thus, verifying the stability of one particular point in the parameters space, and explicit analysis of many such operating points is needed in order to certify stability in some region of parameter values. This makes identification of a stability area unfeasible even for systems with a moderate number of inverters since the number of operating points to check grows exponentially with the increase of the dimension of parameters space (i.e.,, with the increase of the number of inverters). The key advantage of our approach over direct numerical simulations using \eqref{eq:dyn_model} is that the provided regions can be found with much less computational effort, which is almost independent of the number of inverters.

Note that stability boundary given by the criterion $\mu_v = \mu_{cr,min}$ (i.e., when the highest $\mu_v$ crosses the threshold) results in the surface in the space of droop gains $m_i$ having the complicated relationships between each other. To find the closed-form representation, we consider a specific family of the stability regions that provide an upper boundary of $m_i$ for each inverter independently. In such a manner, the inverters could autonomously choose the droop gain from the given region without knowing others' droop gains. Namely, we provide \textit{certified stability area} $\mathcal{G}$ as follows:
    \begin{equation}
        \mathcal{G} = \bigcup_i^v \mathbf{G}_i \ , 
    \end{equation}
where $\mathbf{G}_i$ is the individual certified stability area for each inverter. 

There are multiple ways to construct $\mathbf{G}_i$ regions for each inverter. Specifically, any point of the surface $\mu_v = \mu_{cr,min}$ provides an individual upper boundaries for $m_i$. We start from a particular case corresponding to equal power sharing (i.e., uniform droop gains for all inverters). In this case, all the sets $\mathbf{G}_i$ are the same and can be found as follows:
\begin{equation} \label{eq:stability_region_equal_power_sharing}
    \mathbf{G} =\{(m,n) | \ \frac{m}{5} \leq n \leq \frac{m}{0.3} | \ m \leq \frac{\mu_{cr, min}}{\lambda_{max}(\mathcal{B})} \} ,
\end{equation}
where $\lambda_{max}(\mathcal{B})$ is the maximum eigenvalue of $\mathcal{B}$, and $\mu_{cr, min}=0.826$ - the threshold value for $\mu$. Here we use $\mathbf{G}$ to denote the stability region of every inverter, and $m, n$ are the droop gains of every inverter. Inequalities from \eqref{eq:stability_region_equal_power_sharing} are obtained using $\mu_v = m \lambda_{max}(\mathcal{B})$, which is true because $M = m \mathbf{1}$ for the equal droops. An example of such certified stability region $\mathbf{G}_i$ is given at the Fig.\ref{fig:stability_region_M_N} by the blue triangle (the rest of Fig.\ref{fig:stability_region_M_N} will be discussed in the next section). We note, that all $\mathbf{G}_i$'s represent \emph{convex sets} as they are restricted by linear inequalities, hence their union $\mathcal{G}$ is also a convex set in the multidimensional space of inverters droop gains.

Stability regions in \eqref{eq:stability_region_equal_power_sharing} have been obtained under the assumption of equal power sharing, i.e., equal droop gains for all inverters. However, in practice, it can be more beneficial (and sometimes required due to technical specifications) to realize operation under non-equal droop gains. Although equations \eqref{eq:stability_region_equal_power_sharing} were obtained under equal droop gains assumption, they remain valid even if we assign different droop gains to different inverters, as long as every $m_i$ and $n_i$ remain in the corresponding region $\mathbf{G}_i$. However, certified stability regions $\mathbf{G}_i$ from \eqref{eq:stability_region_equal_power_sharing} are rather conservative for most of the inverters since they are limited by the droop gains of inverters that have the least stability margin (form the \emph{critical cluster}). Any $\mathbf{G}_i$ in \eqref{eq:stability_region_equal_power_sharing} cannot be enlarged without reducing the boundary of at least one other inverter, i.e., \eqref{eq:stability_region_equal_power_sharing} expresses a Pareto frontier. However, the certified stability regions for the non-critical inverters can be significantly enlarged by slightly restricting the stability boundary of the critical inverters. Enlarged $\mathcal{G}$ provides bigger feasibility set for optimization problems. For example, the optimal droop gains for economic dispatch are inversely proportional to the marginal costs \cite{dorfler2015breaking}. Therefore, they are non-equal if the system has various energy sources, and enlarged $\mathcal{G}$ can provide a more economical solution.

To increase the certified stability area $\mathcal{G}$ (through the increase of individual $\mathbf{G}_i$'s), we search for the distribution of droop gains $m_{i,max}$ such that all inverters in the system become similarly critical. If all non-zero eigenvalues $\mu_i$ are equal, then all inverters have the same criticality $\mu_v$. Therefore, to enlarge $\mathcal{G}$, we are to find a set of $m_{i,max}$ such that all $\mu_i$ are equal.
Generally, it is impossible to make all eigenvalues of $C=M\mathcal{B}$ exactly equal by only varying the diagonal matrix of droops $M$. Nevertheless, we could minimize the difference between the largest $\mu_v$ and the smallest non-zero $\mu_2$ as follows:
\begin{equation} \label{eq:optimization}
\begin{aligned}
\min_{M} \quad & \mu_v - \mu_2 \\
\textrm{s.t.} \quad & \mu_v(M\mathcal{B}) = \mu_{cr,min} = 0.826,\\
  &M \succ 0,  \\
  &M \ \text{diagonal}.
\end{aligned}
\end{equation}

Optimization \eqref{eq:optimization} could be performed numerically using the semidefinite programming \cite{shafi2010designing}. While the solution of such an optimization problem can provide a good result, the resulting stability regions  $\mathbf{G}_i$ will be found numerically. 
Here, to get closed-from expressions, we consider another approach based on placing all $\mu_i$ within the same range of values. Specifically, we use the Gershgorin circle theorem to estimate the range of possible values for each $\mu_i$. The theorem says that $\mu_i$ lies within disks centered at $C_{ii}$ with the radius $r_i = \sum_{j\neq i} |C_{ij}| = m_i \sum_{j\neq i} |\mathcal{B}_{ij}| = m_i B_{ii}$. Using Gershgorin disks for each row of $C$ and the fact that $\mu_i$ are real, the following range of values for $\mu_i$ can be derived:

\begin{equation} \label{eq:Gershgoring_for_C}
    0\leq \mu_i \leq 2 m_i \mathcal{B}_{ii}, \ i=1,\cdots,v \ ,
\end{equation}
where $\mathcal{B}_{ii}$ is the $i$-th diagonal element of $\mathcal{B}$. 

To find the values for $m_{i,max}$ two steps are required: first making all the ranges in \eqref{eq:Gershgoring_for_C} equal, and second finding the $m_{i,max}$ corresponding to $\mu_{cr,min}$ for each of them. It is clear that by making all $m_i = \frac{\alpha}{\mathcal{B}_{ii}}$ with some $\alpha > 0$, we can obtain the same boundaries of eigenvalues according to \eqref{eq:Gershgoring_for_C}: $0\leq \mu_i\leq 2\alpha$. Now, it is necessary to define $\alpha$ such that the highest $\mu_v$ is crossing $\mu_{cr,min}$. For doing so, let us define $C_r = \text{diag}(1/\mathcal{B}_{11}, \cdots, 1/\mathcal{B}_{vv}) \mathcal{B}$ then the stability regions $\mathbf{G}_i$ are defined by the following expressions (instead of  \eqref{eq:stability_region_equal_power_sharing}):
\begin{equation} \label{eq:stability_region_non-equal_power_sharing}
    m_i \leq \frac{\mu_{cr,min}}{\lambda_{max}(C_r)\mathcal{B}_{ii}} \ ,
\end{equation}
where $\lambda_{max}(C_r) > 0$ is the highest eigenvalue of $C_r$. It is noteworthy that $\lambda_{max}(C_r) < 2$ according to the Gershgoring circle theorem, which leads to even simpler but more conservative than \eqref{eq:stability_region_non-equal_power_sharing} boundaries:

\begin{equation}\label{eq:stability_regions_final}
    m_i \leq \frac{\mu_{cr,min}}{2\mathcal{B}_{ii}} \ .
\end{equation}
We note, that for both equations \eqref{eq:stability_region_non-equal_power_sharing} and \eqref{eq:stability_regions_final} the boundaries for voltage droop gains $n_i$ are $m_i/5 \leq n_i \leq m_i/0.3$. We also remind that $\mu_{cr,min}$ is a constant (equal to $0.826$). Boundaries given by \eqref{eq:stability_region_non-equal_power_sharing}, in contrast to \eqref{eq:stability_region_equal_power_sharing}, are taking into account the connectedness of inverters as they are inversely proportional to $\mathcal{B}_{ii}$.  Recall that the diagonal element $\mathcal{B}_{ii}$ equals to the sum of admittances of lines connected to the $i$-th node and it is small for weakly connected non-critical inverters and high for strongly connected critical inverters. 

One of the main advantages of \eqref{eq:stability_region_non-equal_power_sharing} and \eqref{eq:stability_regions_final} is that they represent \emph{closed-form} expressions for the stability region of the whole microgrid in the multi-dimensional space of inverter droop gains. Such representation (as compared to numerically calculated boundaries or individual operating points) can be used, for example, to define a feasible domain for optimization problems. As seen from Fig. \ref{fig:stability_region_M_N}, our method has some degree of conservativeness, where only a part of the true stability region can be certified. However, this conservativeness is compensated by the simple analytical forms \eqref{eq:stability_region_equal_power_sharing} or \eqref{eq:stability_region_non-equal_power_sharing} that do not require a lot of computations. We note that the exact stability boundary can not be calculated within a feasible time, even for moderate systems size, as demonstrated in Section \ref{sec:numerical}. Therefore, the exact stability boundary is practically unavailable by direct numerical calculations. 
%----------------------------------------------------%
%               S E C T I O N  V
%----------------------------------------------------% 
\section{Numerical Validation}\label{sec:numerical}

In this section, we demonstrate our methodology on a realistic IEEE 123-bus system with $10$ grid-forming inverters shown in Fig \ref{fig:IEEE123_case2} \cite{schneider2017analytic} and with loads modeled using their effective impedances. There are four types of lines in the network with the $R/X$ ratios of approximately $0.4$, $0.6$, $1.0$, and $2.0$. All the simulations of this section were run on MATLAB software \footnote{\href{https://github.com/goriand/Identification-of-Stability-Regions-in-Inverter-Based-Microgrids}{https://github.com/goriand/Identification-of-Stability-Regions-in-Inverter-Based-Microgrids}}.

First, we illustrate the certified threshold given in \eqref{eq:stability_region_equal_power_sharing}.
 To enable 2-$D$ visualisation, we have assumed uniform droop gains of the inverters. Fig. \ref{fig:stability_region_M_N} shows the resulting  stability regions as the function of (uniform) voltage and frequency droop gains. The blue region corresponds to the certified stability region $\mathcal{G}$ defined by \eqref{eq:stability_region_equal_power_sharing} as $m_i < 5.29\%$. Note that this region provides the upper limit for only frequency droops $m_i$ resulting in the vertical line in Fig. \ref{fig:stability_region_M_N_non-equal_power_sharing}.
The curved boundary line of the green region corresponds to the actual stability boundary, calculated using \eqref{eq:dyn_model}.
Namely, for each $k\in[0.3,5]$ the boundary value of $m_i$, for which $\Re{\lambda} = 0$, was calculated using binary search. Gray sectors represent regions outside the assumed range for $k = \frac{m}{n} \in [0.3,5]$ such that the analyzed region lies in between lines $m = 0.3 n$ and $m = 5 n$.  

\begin{figure}
    \centering
    \includegraphics[width=0.4\textwidth]{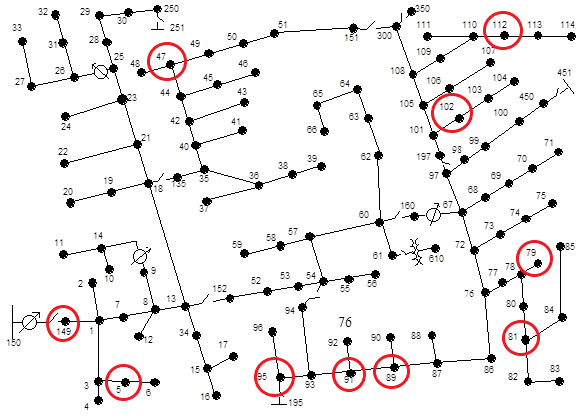}
    \caption{IEEE $123$-bus system. Inverters are located at nodes: $95, 149, 79, 5, 102, 112, 81, 91, 89, 47$ (denoted by red circles).}
    \label{fig:IEEE123_case2}
\end{figure}

\begin{figure}
    \centering
    \includegraphics[width=0.4\textwidth]{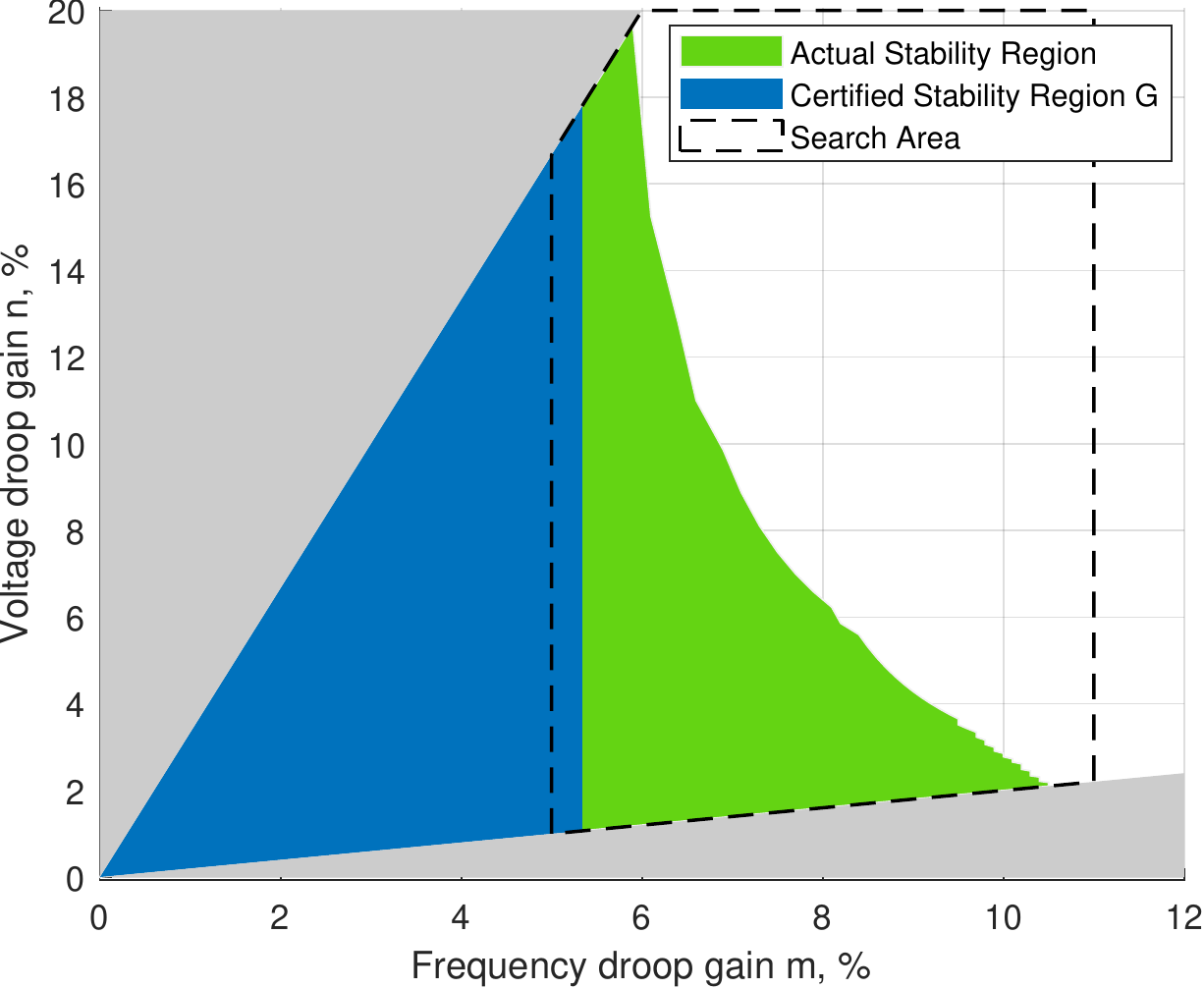}
    \caption{Stability regions under equal droop gains assumption. The certified stability region $\mathbf{G}$ \eqref{eq:stability_region_equal_power_sharing} is given in blue, while the true stability boundary (that is possible to find under the equal droop gains assumption) is in green.}
    \label{fig:stability_region_M_N}
\end{figure}

\begin{figure}
    \centering
    \includegraphics[width=0.4\textwidth]{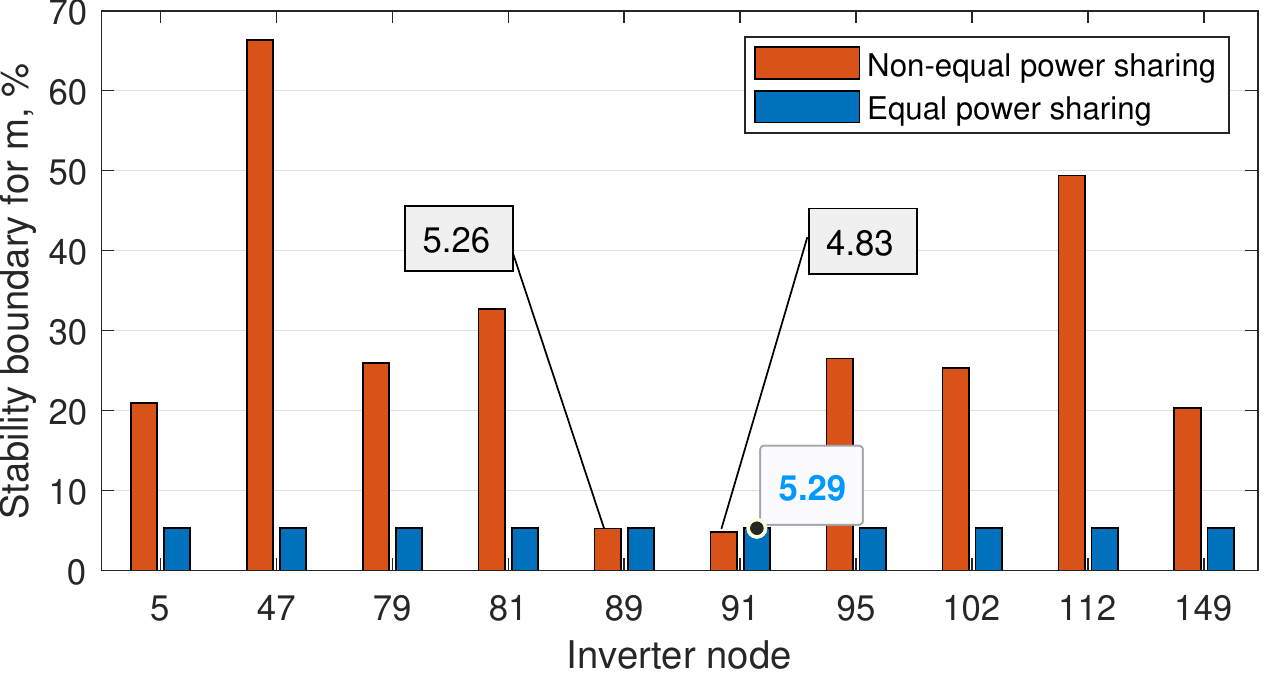}
    \caption{Maximum frequency droop gains of individual inverters $m_{i,max}$ (corresponding to certified stability regions $\mathbf{G}_i$) with equal and non-equal droop gains assumption according to \eqref{eq:stability_region_equal_power_sharing} and \eqref{eq:stability_region_non-equal_power_sharing} respectively. 
    }
    \label{fig:stability_region_M_N_non-equal_power_sharing}
\end{figure}

Next, let us consider the construction of $\mathcal{G}$ based on \eqref{eq:stability_region_non-equal_power_sharing} instead of \eqref{eq:stability_region_equal_power_sharing}, i.e., without assuming equal droop gain. In this case, for each inverter there is a certified triangle $\mathbf{G}_i$, similar to the blue triangle of Fig. \ref{fig:stability_region_M_N}, bounded by the same range of slopes $k=m/n \in [0.3, 5]$ but with a different value of the maximum frequency droop gain $m_{i,max}$ for each inverter. Fig. \ref{fig:stability_region_M_N_non-equal_power_sharing} provides the values of $m_{i,max}$ for this case, and also, for comparison, for the case of equal droop gains.  
Note that it is impossible to show the actual stability boundaries, similarly as in Fig. \ref{fig:stability_region_M_N}, without searching the possible solution space as droop gains are no longer uniform.
The suggested $\mathcal{G}$ with non-equal power sharing increases the range of values of feasible (stable) gains and therefore provides more flexibility in tuning them. The cost we pay for enlarging the region $\mathcal{G}$ is a slightly reduced gain of the critical inverters in nodes 89 and 91. However, the certified stability regions $G_i$ for other inverters can be significantly enlarged. For instance, it is enlarged by more than $10$ times for inverter at node $47$ compared with equal power sharing case \eqref{eq:stability_region_equal_power_sharing}. At the same time, the $G_i$ for the critical inverter at node $89$ is reduced by only $\sim 10\%$ using \eqref{eq:stability_region_non-equal_power_sharing}.

To verify that $\mathcal{G}$ indeed corresponds to stability region, in Fig. \ref{fig:pole-zero} we show an eigenvalue plot comparing the eigenvalues of the system \eqref{eq:dyn_model} with each $m_{i,max}$ equal to the upper limits of \eqref{eq:stability_region_non-equal_power_sharing} and \eqref{eq:stability_region_equal_power_sharing}. The values of voltage droop gains $n_{i,max} = m_{i,max}/0.3$, which correspond to the maximum values of $m_i, n_i$ inside $\mathbf{G}_i$. All poles have negative real parts that show stability of the chosen point of $\mathcal{G}$.

In addition, to validate the worst-case condition discussed in Section \ref{sec:identification_procedure}, eigenvalues for a system \eqref{eq:dyn_model} with various non-uniform $R/X$ and $k$ are provided in Fig. \ref{fig:IEEE123_eigenplot_nonequal_k_rho_var}. Specifically, the eigenvalues for s system with worst-case uniform parameters $\rho_{ext} = 1.3, k_{ext} = 0.3$ (green and black crosses) are compared with eigenvalues of $500$ randomly sampled systems with non-uniform $\rho, k$ for both cases of equal \eqref{eq:stability_region_equal_power_sharing} and non-equal \eqref{eq:stability_region_non-equal_power_sharing} droops (purple and orange dots). All purple and orange dots are in the left-half plane comparing with marginally stable crosses lying on the imaginary axis, i.e., any deviation from $\rho_{ext}, k_{ext}$ stabilizes the system.
This confirms the validity of our methodology for the systems with non-uniform values of $\rho$ and $k$.
\begin{figure}
    \centering
    \includegraphics[width=0.4\textwidth]{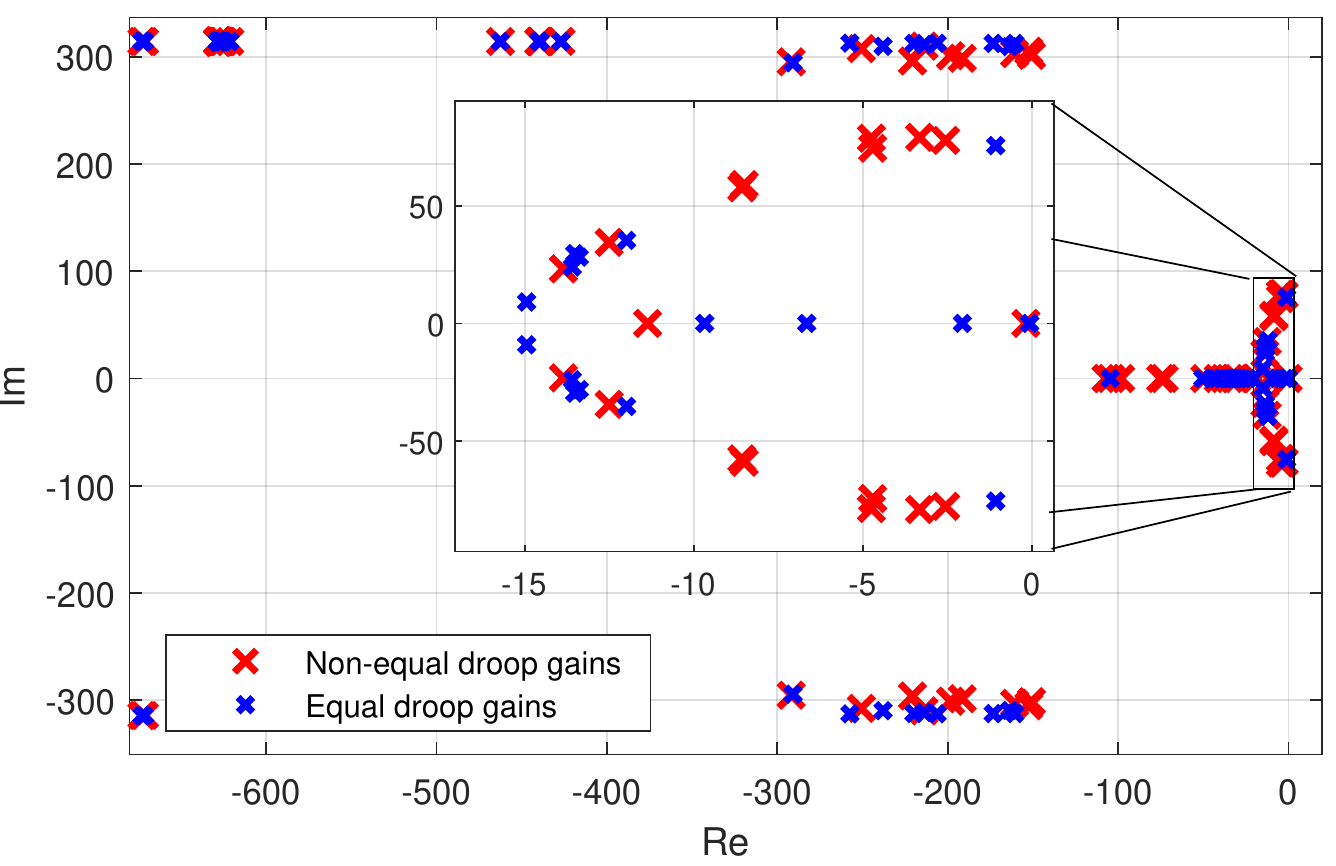}
    \caption{Eigenvalue plot for the system of Fig. \ref{fig:IEEE123_case2} (with full dynamic model \eqref{eq:dyn_model}) for the case $m_i = m_{i,max}$ according to \eqref{eq:stability_region_equal_power_sharing} (blue crosses) and \eqref{eq:stability_region_non-equal_power_sharing} (red crosses).}
    \label{fig:pole-zero}
\end{figure}

\begin{figure}
    \centering
    \includegraphics[width=0.45\textwidth]{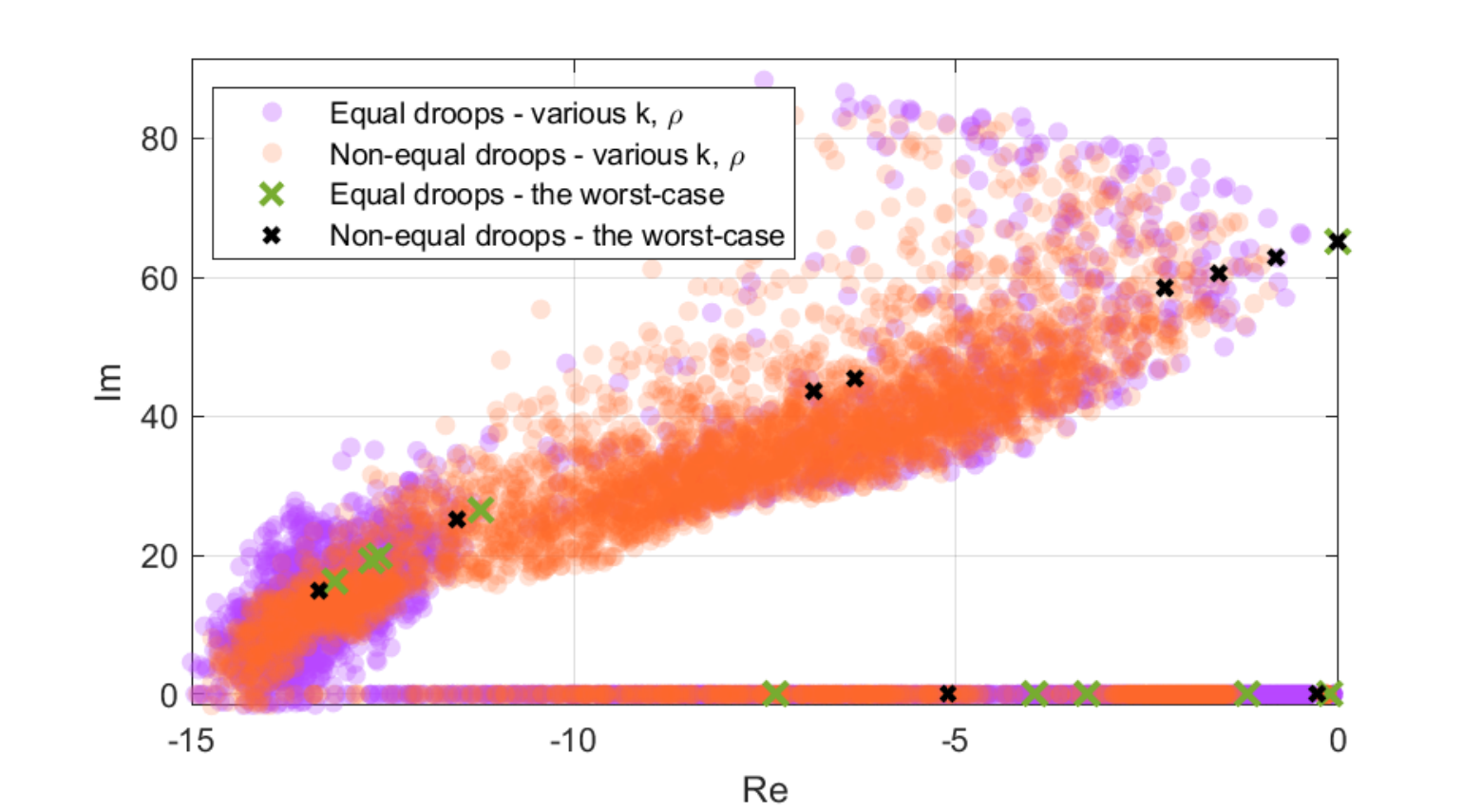}
    \caption{Eigenvalue plot for the system of Fig. \ref{fig:IEEE123_case2} with variable (and non-uniform) $k, \rho$ for equal \eqref{eq:stability_region_equal_power_sharing} and non-equal \eqref{eq:stability_region_non-equal_power_sharing}  droop cases. Only the dominant modes of an upper-half plane are shown. The eigenvalues for the non-equal droops case are shown in orange, while the eigenvalues for the equal droops case are purple (orange dots are superimposed on purple dots). All dots are transparent so that darker areas correspond to a higher density of eigenvalues.}
    \label{fig:IEEE123_eigenplot_nonequal_k_rho_var}
\end{figure}

\subsection{Computational Complexity}
As was shown in the previous section, our methodology allows finding approximate stability regions for microgrids with an arbitrary number of inverters - for instance,  the blue area in Fig. \ref{fig:stability_region_M_N} determines a range of gain values of certified (although conservatively) stable points. Although every individual operating point in our region could possibly be checked by direct eigenvalue analysis of the system \eqref{eq:dyn_model}, construction of stability regions using such a point-by-point approach becomes computationally prohibitively expansive even for systems with a modest number of inverters.
To illustrate this, we estimate the number of operating points that need to be checked to find the true stability boundary (green region at Fig.\ref{fig:stability_region_M_N}). Note that we could find the actual stability limits in Fig. \ref{fig:stability_region_M_N} by searching the solution space because there were only two variables: $m, n$. Suppose that the goal is to search within the dashed polygon bounded by gray areas and $n_i\leq 20\%$, $5\%\leq m_i\leq 11\%$ (we assume we have already excluded from consideration the region of small values of droop gains to reduce the computational burden). 
Thus, the area in two dimensions of droop gains for every inverter to be covered by the search is $A_i \sim (11-5)5(\frac{1}{0.3} - \frac{1}{5}) \sim 10^2$ (we approximately assume the search area to be rectangular). The resulting volume in the full $2v$ dimensional space of inverter droop gains is $(A_i)^v$. Considering the mesh discretization step for each droop gain to be $\epsilon$, we can provide the lower bound of the number of points required as follows,

\begin{equation}\label{eq:num_of_points}
    N_p = \frac{(A_i)^v}{\epsilon^{2v}} \sim (\frac{10}{\epsilon})^{2v} . 
\end{equation}

This result is a classic example of the \emph{curse of dimensionality} when sampling in multi-dimensional space. For example, when $v=10$ and $\epsilon = 0.5\%$ (i.e., for a system with $10$ inverters and with a mesh discretization of $0.5\%$ for droop gains), the total number of operating points to check is $N_p \sim 10^{25}$. Although we used a rather coarse-grain estimation assuming a simple choice of the mesh, the resulting number is exceptionally large. Even if some more sophisticated methods are used to perform the sampling of operating points to check, reducing the complexity by several orders of magnitude, the resulting number of points is still way too large for any practical purpose. Table \ref{tab:comp_complexity} provides an explicit indication of the estimation of the computation time required for the systems of $4$ and $10$ inverters, respectively. Even for the former one, the computation time is very high. However, for the system of $10$ inverters, the computation complexity is so high that even if very advanced methods for eigenvalues calculation are used, it is still non-feasible to determine stability boundary by direct numerical simulations.  

On the other hand, using our proposed method, the calculation of suggested stability regions in \eqref{eq:stability_region_non-equal_power_sharing} requires only the calculation of the eigenvalue $\lambda_{max}(C_r)$, which corresponds to quadratic computation complexity $O(v^2)$. The comparison of computational times for our method and direct numerical simulations is provided in Table \ref{tab:comp_complexity}. The first two rows correspond to computation times averaged among $100$ measurements using MATLAB software run on an Intel Core i5-4690 Processor at 3.5 GHz. In contrast, the computation times in the last row are calculated based on the obtained estimation of the number of points $N = A_i^v/\epsilon^{2v} = (20)^{2v}$. Table \ref{tab:comp_complexity} clearly demonstrates that the direct computation even for a moderate two-area system of Fig. \ref{fig:kundur_decoupling} with $4$ inverters is already practically infeasible, while the time required to calculate stability region for a $10$-inverter IEEE 123-bus system is far beyond any practical reach.

\begin{table}[]
    \centering
    \newcommand{\wrap}[1]{\parbox{.26\linewidth}{\vspace{1.5mm}#1\vspace{1mm}}}
    \caption{Computational Time for Stability Assessment Using the Proposed Methods vs. Direct Numerical Approach}
    \begin{tabular}{|c||c|c|}
    \hline
          &  \wrap{Two-area system of Fig. \ref{fig:kundur_decoupling} ($4$ inverters)} & \wrap{IEEE 123-bus system of Fig. \ref{fig:IEEE123_case2} ($10$ inverters)}\\
         \hline
    \wrap{Calculation of stability region  \eqref{eq:stability_region_non-equal_power_sharing}} & $0.55$ ms & $2.40$ ms \\
    \hline
        \wrap{Calculation of eigenvalues of \eqref{eq:dyn_model} (single operating point)} & $0.14$ ms & $530$ ms \\
        \hline
        \wrap{Direct stability boundary calculation (estimated according to \eqref{eq:num_of_points} )} & $1000$ hours & $10^{18}$ years \\
        \hline
    \end{tabular}
    \label{tab:comp_complexity}
\end{table}

\section{Conclusion}

In this paper, a novel technique for stability assessment of inverter-based microgrids has been proposed that allows the construction of \emph{certified stability regions} in the multidimensional space of inverter droop gains. Contrary to the direct point-by-point approach, the numerical complexity of our method does not grow significantly with the increase in the number of inverters in the system. The method is based on the decomposition of a system into a set of two-bus equivalents through the use of the generalized Laplacian matrix $C$. We then have demonstrated that the eigenvalues of $C$ could be used as the metric of \textit{criticality}, i.e., the proximity of a group of inverters to the stability boundary. We exploited this metric to provide certified stability region $\mathcal{G}$ for a system with an arbitrary number of inverters. This region is a unification of certified stability regions $\mathbf{G}_i$ for individual inverter's droop gains, each of which represents a convex set. Therefore, the unified region $\mathcal{G}$ is also a convex set in the multidimensional space of the droop gains of all inverters, which gives significant advantages for its use in gain optimization problems. 

The proposed methodology requires far less computational effort than the direct point-by-point simulation so that stability regions in high dimensional space of droop gains can be determined.  
The certified stability regions can be conveniently calculated and visualized. The method has been tested numerically using IEEE $123$-bus system parameters with $10$ grid-forming inverters. The simulations confirmed that the calculated approximate stability boundaries are within the true stability region, which we verify by applying the full dynamic model. 

\ifCLASSOPTIONcaptionsoff
  \newpage
\fi

\bibliographystyle{IEEEtran}
\bibliography{bibl}

 \appendices
 
 \section{Elimination of the passive loads}\label{app:elimination_loads}

After Kron reduction of virtual buses, only buses with inverters and loads remain. So let us now divide the remaining buses into two sets: outputs of inverters labeled as `$O$' and load injections `$L$'.
Any load connected directly to an inverter terminal is moved to a new artificial node connected to the original inverter terminal by an infinitely low impedance. Consequently, the sets $O$ and $L$ do not intersect, and the state space matrix is divided into four blocks accordingly:

\begin{equation}
    A = \left[
    \begin{array}{c|c}
        A_{OO} & A_{OL} \\ 
        \hline
        A_{LO} & A_{LL}
    \end{array}
    \right] ,
\end{equation}
where block $A_{OO}$ is associated with the states of inverters, $\pmb{x}_O = (\pmb{\theta}, \pmb{\omega}, \pmb{V}, \pmb{I}_d, \pmb{I}_q)^T$, and has the same form as in \eqref{eq:5th_model}, but instead of the susceptance matrix with Kron-reduced load nodes $\mathcal{B} = \mathcal{B}_{OO} - \mathcal{B}_{OL}\mathcal{B}_{LL}^{-1}\mathcal{B}_{LO}$, there is just $\mathcal{B}_{OO}$ associated with inverter buses only. Other $A$ blocks are related with the states (current injections) of loads, $\pmb{x}_L = (\{\pmb{I}_d\}_L, \{\pmb{I}_q\}_L)^T$, and are listed as follows,

\begin{subequations}
    \begin{equation}
        A_{OL} = \begin{bmatrix}
            0 & 0 \\
            0 & 0 \\
            0 & 0\\
            \mathcal{B}_{OL}R_{L} & -\mathcal{B}_{OL}X_{L} \\
            \mathcal{B}_{OL}X_{L} & \mathcal{B}_{OL}R_{L}\\
            \end{bmatrix} \ ,
    \end{equation}
   where $R_L$ and $X_L$ are diagonal matrices of load impedances,
    \begin{equation}
        A_{LO} = \begin{bmatrix}
                0&0&\mathcal{B}_{LO}&0&0\\
                \mathcal{B}_{LO}&0&0&0&0
                \end{bmatrix} \ ,
    \end{equation}
    \begin{equation} \label{eq:A_ll}
        \begin{split} 
            &A_{LL} = \begin{bmatrix}
                        \mathcal{B}_{LL} R_{L} - \rho \mathbf{1} & -\mathcal{B}_{LL} X_{L} + \mathbf{1}\\
                        \mathcal{B}_{LL} X_{L} - \mathbf{1} & \mathcal{B}_{LL} R_{L} - \rho \mathbf{1}
                        \end{bmatrix} = \\ 
                        &
                    \begin{bmatrix}
                        \mathcal{B}_{LL}  & 0 \\
                        0  & \mathcal{B}_{LL}
                        \end{bmatrix}    
                    \begin{bmatrix}
                        R_{L} - R_{LL} & - X_{L} + X_{LL}\\
                         X_{L} - X_{LL} & R_{L} - R_{LL}
                        \end{bmatrix} \ ,
        \end{split}
    \end{equation}
\end{subequations}
where $R_{LL} = \rho \mathcal{B}_{LL}^{-1} = \Re{Y_{LL}^{-1}}$, $X_{LL} = \mathcal{B}_{LL}^{-1} = \Im{Y_{LL}^{-1}}$ are nodal resistance and nodal reactance matrices corresponding to buses with loads (not including load impedances $R_L, X_L$).In essence, the right-hand side of \eqref{eq:A_ll} using the complex impedance matrices could be represented as $\mathcal{B}_{LL}[Z_L - Z_{LL}]$. Further, $Z_{LL}$ is neglected using the fact that load impedance is much higher than impedances of the network lines, resulting in the following approximation:

\begin{equation}
    A_{LL} \approx \begin{bmatrix}
                        \mathcal{B}_{LL}  & 0 \\
                        0  & \mathcal{B}_{LL}
                        \end{bmatrix}    
                    \begin{bmatrix}
                        R_{L}  & - X_{L}\\
                         X_{L} & R_{L}
                        \end{bmatrix} \ .
\end{equation}

By assuming the quasi-stationary approximation for loads, the load current injections could be excluded from the state-space model as follows:

\begin{equation}
    A_{approx} = A_{OO} - A_{OL}A_{LL}^{-1}A_{LO} \ .
\end{equation}
Finally, substituting all the expressions for $A$ blocks one obtains the following:
\begin{equation}\label{eq:A_approx_complement}
\begin{split}
    &A_{OL}A_{LL}^{-1}A_{LO} = \\
    &\begin{bmatrix}
        0 & 0 &0 &0 &0\\
        0 &0 &0 &0 &0\\
        0 &0 &0 &0 &0\\
        0 &0 & \mathcal{B}_{OL}\mathcal{B}_{LL}^{-1}\mathcal{B}_{LO} &0 &0\\
        \mathcal{B}_{OL}\mathcal{B}_{LL}^{-1}\mathcal{B}_{LO} &0 &0 &0 &0
    \end{bmatrix} \ .
\end{split}
\end{equation}
Consequently, the derived $A_{approx}$ does not include $R_L, X_L$ and coincide with \eqref{eq:5th_model} because non-zero elements of \eqref{eq:A_approx_complement} complement $A_{OO}$ to the Kron-reduced susceptance matrix as $\mathcal{B} = \mathcal{B}_{OO} - \mathcal{B}_{OL}\mathcal{B}_{LL}^{-1}\mathcal{B}_{LO}$. That concludes our derivation. Note that the resulting system \eqref{eq:5th_model} is unloaded.

  \section{Stationarity with Non-homogeneous $R/X$ }\label{app:stationarity_proof}

\begin{lemma} \label{lemma:G_representation}
    Assuming that the network has been Kron-reduced, the dynamic model \eqref{eq:dyn_model} in terms of $(\pmb{\mathcal{I}_d}, \pmb{\mathcal{I}_d})$ has the following form in the Laplace domain:
    
    \begin{subequations} \label{eq:dyn_model_Idq}
    \begin{equation}
        G(s) \begin{pmatrix}
          \pmb{\mathcal{I}_d} \\ \pmb{\mathcal{I}_q}
        \end{pmatrix} = 0 \ ,
    \end{equation}
    \begin{equation}
        G(s) =  \begin{bmatrix}
            \tau_0 s I + P & -I - \frac{X^{-1}\nabla N \nabla^T}{\tau s + 1} \\
            I + \frac{X^{-1}\nabla M \nabla^T}{s(\tau s + 1)} \ & \tau_0 s I + P
        \end{bmatrix}
         \ .
    \end{equation}
       
    \end{subequations}
\end{lemma}
\begin{proof}
 One could substitute $\pmb{\theta} = - \frac{M}{\tau s + 1} \pmb{P} = - \frac{M}{\tau s + 1} \pmb{I_d} = - \frac{M}{\tau s + 1} \nabla^T \pmb{\mathcal{I}_d} $ and $\pmb{V} = - \frac{N}{\tau s + 1} \pmb{Q} =  \frac{N}{\tau s + 1} \pmb{I_q} = \frac{N}{\tau s + 1} \nabla^T \pmb{\mathcal{I}_q}$ into \eqref{eq:Ohms_law} to obtain the desired representation.
\end{proof}

The eigenvalues $\lambda$ of the model \eqref{eq:dyn_model_Idq} (and \eqref{eq:dyn_model} as well) correspond to a non-trivial solution $\lambda, \pmb{\phi}_d, \pmb{\phi}_q$ for \eqref{eq:dyn_model_Idq} obtained by replacing $s$ with $\lambda$, $\pmb{\mathcal{I}_d}$ with $\pmb{\phi}_d$ and $\pmb{\mathcal{I}_q}$ with $\pmb{\phi}_d$. We call $\pmb{\phi} = \begin{pmatrix} \pmb{\phi}_d\\ \pmb{\phi}_q \end{pmatrix}$ the right eigenvector or just eigenvector for the polynomial eigenvalue problem \eqref{eq:dyn_model_Idq}. We also define the left eigenvector $\pmb{\psi} = \begin{pmatrix} \pmb{\psi}_d\\ \pmb{\psi}_q \end{pmatrix}$ as the right eigenvector of the (Hermitian) transposed matrix polynomial of \eqref{eq:dyn_model_Idq}.

Using the model representation in terms of currents $(\mathcal{I}_d, \mathcal{I}_q)$ given in Lemma \ref{lemma:G_representation} below, we express the $\lambda$ sensitivity as follows,

\begin{equation} \label{eq:sensitivity_G}
    \frac{\partial \lambda}{\partial \rho_j}\Bigr|_{\rho_{cr}} =
    - \frac{\begin{pmatrix}\pmb{\psi}_d^\dagger & \pmb{\psi}_q^\dagger\end{pmatrix} \frac{\partial G}{\partial \rho_j}|_{\rho_{cr}}  \begin{pmatrix}\pmb{\phi}_d \\ \pmb{\phi}_q\end{pmatrix}}{\begin{pmatrix}\pmb{\psi}_d^\dagger & \pmb{\psi}_q^\dagger\end{pmatrix} \frac{\partial G}{\partial \lambda}|_{\rho_{cr}} \begin{pmatrix}\pmb{\phi}_d \\ \pmb{\phi}_q\end{pmatrix}} \ .
\end{equation}

We notice that the denominator of \eqref{eq:sensitivity_G} is independent of $j$. Let us analyze the numerator of \eqref{eq:sensitivity_G}:
\begin{subequations}
    \begin{equation}
        \mathcal{N}_j \overset{\Delta}{=} \begin{pmatrix}\pmb{\psi}_d^\dagger & \pmb{\psi}_q^\dagger\end{pmatrix} \frac{\partial G}{\partial \rho_j}\Bigr|_{\rho_{cr}}  \begin{pmatrix}\pmb{\phi}_d \\ \pmb{\phi}_q\end{pmatrix}
    \end{equation}
    \begin{equation}
        \mathcal{N}_j = \lambda(\tau \lambda + 1) ([\psi_d]_j^*[\phi_d]_j + [\psi_q]_j^*[\phi_q]_j ) \ .
    \end{equation}
    
\end{subequations}

Further, by straightforward but cumbersome manipulations we derive:

\begin{equation}
    \mathcal{N}_j = - 2 X_j |[\phi_d]_j|^2 c   \frac{\lambda(\tau\lambda+1) + \tilde{\mu}}{\tau_0\lambda+\rho}  \ ,
\end{equation}
where $\tilde{\mu}$ is the eigenvalue of $X^{-1}\nabla M \nabla^T$, $c = \frac{\pmb{\psi}^\dagger \pmb{\psi}}{\pmb{\psi}_d^T X \pmb{\phi}_q + \pmb{\psi}_q^T X \pmb{\phi}_d} e^{j2\varphi_d}$. The phase $\varphi_d$ is the same for each element of $\pmb{\phi}_d$ as it the eigenvector of generalized Laplacian matrix and could be choosen to be real.
Now, we see that desired $\alpha_j = X_j |[\phi_d]_j|^2$. That concludes our proof.

\end{document}